\documentclass[abstract,headings=normal,DIV=13]{scrartcl}

\usepackage[automark]{scrpage2}
\usepackage[english]{babel}
\usepackage[T1]{fontenc}
\usepackage[applemac]{inputenc}
\usepackage{graphicx}
\usepackage{tikz}
\usepackage{amsmath}
\usepackage{amssymb}
\usepackage{amsthm}
\usepackage{subfig}
\usepackage{accents}
\usepackage{url}
\usepackage[colorlinks=true,linkcolor=blue,citecolor=red,urlcolor=black]{hyperref}

\usetikzlibrary{arrows}

\addtokomafont{caption}{\small}

\newcommand{\Z}{\mathbb{Z}}

\newcommand{\Ell}{\mathcal{L}}
\newcommand{\ELL}{\mathfrak{L}}
\newcommand{\cL}{\mathcal{L}}

\newcommand{\M}{\mathfrak{M}}

\newcommand{\wx}{\widetilde{x}}
\newcommand{\wip}{\widetilde{p}}
\newcommand{\whx}{\widehat{x}}

\newcommand{\htilde}[1]{\widehat{\widetilde{#1}}}

\makeatletter
\renewcommand*\env@cases[1][1.2]{%
  \let\@ifnextchar\new@ifnextchar
  \left\lbrace
  \def\arraystretch{#1}%
  \array{@{}l@{\quad}l@{}}%
}
\makeatother

\newtheorem{theo}{Theorem}[section]

\newtheorem{prop}[theo]{Proposition}

\theoremstyle{definition}
\newtheorem{defi}[theo]{Definition}
\theoremstyle{remark}
\newtheorem*{rem}{Remark}

\graphicspath{{./graphics/}}
\DeclareGraphicsExtensions{.pdf}

\begin{document}

\title{Multi-time Lagrangian 1-forms \\ for families of B\"acklund transformations.\\
Relativistic Toda-type systems}
\author{Raphael~Boll\and Matteo~Petrera\and Yuri~B.~Suris}
\publishers{\vspace{0.5cm}{\small Institut f\"ur Mathematik, MA 7-2, Technische Universit\"at Berlin,\\
Stra{\ss}e des 17. Juni 136, 10623 Berlin, Germany\\
E-mail: \url{boll}, \url{petrera}, \url{suris@math.tu-berlin.de}}}
\maketitle
\begin{abstract}
\noindent We establish the pluri-Lagrangian structure for families of B\"acklund transformations of relativistic Toda-type systems. The key idea is a novel embedding of these discrete-time (one-dimensional) systems into certain two-dimensional pluri-Lagrangian lattice systems. This embedding allows us to identify the corner equations (which are the main building blocks of the multi-time Euler-Lagrange equations) with local superposition formulae for B\"acklund transformations. These superposition formulae, in turn, are key ingredients necessary to understand and to prove commutativity of the multi-valued B\"acklund transformations. Furthermore, we discover a two-dimensional generalization of the spectrality property known for families of B\"acklund transformations. This result produces a family of local conservations laws for two-dimensional pluri-Lagrangian lattice systems, with densities being derivatives of the discrete 2-form with respect to the B\"acklund (spectral) parameter. Thus, a relation of the pluri-Lagrangian structure with more traditional integrability notions is established.
\end{abstract}

\section{Introduction}
This paper can be considered as a continuation of our recent paper \cite{nonrel}, where we gave an application of the general Lagrangian theory of discrete integrable systems of classical mechanics, developed in \cite{S12}, to families of B\"acklund transformations for non-relativistic Toda-type systems. The development of the general theory in \cite{S12} was prompted by an example of the discrete time Calogero-Moser system studied in \cite{YLN}, and belongs to the line of research on Lagrangian theory of discrete integrable systems initiated by Lobb and Nijhoff in \cite{LN} and followed by a number publications \cite{LNQ,LN2,BS1,ALN,lagrangian,pluriharmonic,variational,octahedron}. The notion of integrability of discrete systems, lying at the basis of this development, is that of the multidimensional consistency. This understanding of integrability has been a major breakthrough \cite{quadgraphs,N}, and stimulated an impressive activity boost in the area, cf.~\cite{DDG}.

The original idea of Lobb and Nijhoff can be summarized as follows: solutions of integrable systems deliver critical points simultaneously for actions along all manifolds of the corresponding dimension in multi-time; the Lagrangian form is closed on solutions. This idea resembles the classical notion of pluriharmonic functions and, more generally, of pluriharmonic maps \cite{R, OV, BFPP}, which are simultaneous extremals of the Dirichlet energy along all holomorphic curves in a multi-dimensional complex vector space. This motivated us in \cite{variational, pluriharmonic} to introduce a novel term for the new branch of the calculus of variations specific for integrable systems: we call the corresponding systems {\em pluri-Lagrangian}, and we argue that integrability of variational systems should be understood as the existence of the pluri-Lagrangian structure. In the present paper, we hope to provide an additional evidence in favor of this view.

Here, we establish and investigate the pluri-Lagrangian structure for a more general class of systems than the one studied in \cite{nonrel}, namely for the so-called relativistic Toda-type systems. The general form of a relativistic Toda-type system is
\begin{equation*}\label{eq:RTL}
\ddot{x}_k=r(\dot{x}_k)(f(x_{k+1}-x_k)-f(x_k-x_{k-1})+\dot{x}_{k+1}g(x_{k+1}-x_k)-
\dot{x}_{k-1}g(x_k-x_{k-1})).
\end{equation*}
The general form of a discrete-time relativistic Toda-type system is
\begin{equation*}\label{eq:dRTL}
G(\widetilde{x}_k-x_k)-G(x_k-\undertilde{x}_k)=H(x_{k+1}-x_k)-H(x_k-x_{k-1})
+F(\undertilde{x}_{k+1}-x_k)-F(x_k-\widetilde{x}_{k-1}).
\end{equation*}
A theory and an exhaustive list of integrable systems of this type can be found in \cite{ASh1, ASh2, A, S, Todapaper}. To derive a pluri-Lagrangian structure for these (one-dimensional) systems, it turned out to be necessary to re-interpret them as a particular case of two-dimensional lattice systems, and to develop a general theory of discrete two-dimensional pluri-Lagrangian systems. The latter goal was achieved in \cite{variational}.

The structure and the main results of the present paper are as follows.
\begin{itemize}
\item In section \ref{sect: discr results}, we provide the reader with an overview of the theory of pluri-Lagrangian systems in dimensions $d=1,2$, following mainly \cite{S12, variational}. The fundamental notion of consistent systems of 2D, resp. 3D corner equations, which are main building blocks of pluri-Lagrangian systems, is reminded in detail.
\item Then, in section \ref{sect: from pluri to rel}, we present the construction of two mutually commuting families of symplectic maps (B\"acklund transformations) from a generic discrete two-dimensional pluri-Lagrangian system generated by a discrete three-point 2-form. The commutativity proof is based on the construction of the so called {\em local superposition formulae} which turn out to be nothing but the suitably interpreted 3D corner equations. Moreover, these superposition formulae enable us to handle the multi-valuedness of B\"acklund transformations in the case of periodic boundary conditions, by means of a precise description of the branching behavior of the multi-valued maps.
\item In sections \ref{sect: sym}--\ref{sect: rat}, these general results are applied to all systems of the relativistic Toda type, as listed in \cite{A, S}. For each of the systems, we identify all the ingredients of the pluri-Lagrangian structure, which allows us to give unified proofs for commutativity of all maps in question. In particular, in all cases we prove the so-called \emph{closure relation}, which expresses the fact that the Lagrangian 1-form on the multi-time (space of independent variables) is closed on solutions of variational equations, and turns out to be the main feature of the Lagrangian theory. These results generalize the ones from our recent paper~\cite{nonrel}: sending the relativistic parameter $\alpha$ to $0$, we recover the corresponding results for the non-relativistic case obtained in~\cite{nonrel}. This reinforces the observation already made in \cite{Todapaper}: the non-relativistic degeneration obscures the natural relations to two-dimensional lattice systems.
\item Finally, in section \ref{sect: spectrality}, we turn to the question of the relation of the pluri-Lagrangian structure to more traditional attributes and notions of integrability. In the one-dimensional context, such a relation is established through connecting the closure relation with the \emph{spectrality property}, introduced by Kuznetsov and Sklyanin \cite{KS}, which says that the derivative of the Lagrangian with respect to the parameter of a family of B\"acklund transformations is a generating function of common integrals of motion for the whole family. In the two-dimensional context, we establish a new result which connects the closure relation with a parameter-dependent family of {\em local conservation laws}. Again, the densities of these conservation laws turn out to be composed of the derivatives of the Lagrangian 2-form with respect to the B\"acklund parameter. Moreover, in the framework of the relativistic Toda-type systems, these conservation laws turn out to constitute a local form of the integrals of motion provided by the spectrality property.
\end{itemize}

\section{General theory of discrete pluri-Lagrangian systems}\label{sect: discr results}

\begin{defi}[$d$-dimensional pluri-Lagrangian problem]\label{def:pluriLagr problem}
Let $\Ell$ be a discrete $d$-form on $\Z^{m}$, depending on some field $x:\Z^{m}\to\mathcal{X}$, where $\mathcal{X}$ is some vector space.
\begin{itemize}
\item To an arbitrary oriented $d$-dimensional manifold $\Sigma$ in $\Z^{m}$, there corresponds the \emph{action functional}, which assigns to $\left.x\right|_{V(\Sigma)}$, i.e., to the fields at the vertices of $\Sigma$, the number
\[
S_{\Sigma}=\sum_{\sigma\in\Sigma}\Ell(\sigma).
\]
\item We say that the field $x:V(\Sigma)\to\mathcal{X}$ is a \emph{critical point} of $S_{\Sigma}$, if at any interior point $n\in V(\Sigma)$, we have
\[\frac{\partial S_{\Sigma}}{\partial x(n)}=0.
\]
\item We say that the field $x:\Z^{m}\to\mathcal{X}$ solves the \emph{pluri-Lagrangian problem} for the Lagrangian $d$-form $\Ell$ if, \emph{for any oriented $d$-dimensional manifold $\Sigma$ in $\Z^{m}$}, the restriction $\left.x\right|_{V(\Sigma)}$ is a critical point of the corresponding action $S_{\Sigma}$.
\end{itemize}
\end{defi}

\subsection{One-dimensional pluri-Lagrangian systems, \texorpdfstring{$d=1$}{d=1}}\label{one-dimensional}

This section is based on \cite{S12}.

In the case $d=1$, $\Ell$ is a function of directed edges $\sigma$ of $\Z^{m}$ with $\Ell\left(-\sigma\right)=-\Ell\left(\sigma\right)$. Thus,
\[
\Ell(\sigma_i)=\Ell(n,n+e_i)=\Lambda_i(x,x_i)\quad\Leftrightarrow\quad \Ell(-\sigma_i)=\Ell(n+e_{i},n)=-\Lambda_i(x,x_{i}).
\]
Here $\Lambda_i:{\mathcal X}\times {\mathcal X}\to\mathbb R$ are local Lagrangian functions corresponding to the edges of the $i$\textsuperscript{th} coordinate direction, and the following abbreviations are used: $x$ for $x(n)$ at a generic point $n\in\Z^m$, and then
\begin{alignat}{3}\label{eq: shifts}
    &x_i=x(n+e_i),&\qquad& x_{-i}=x(n-e_i),&\qquad& i=1,\ldots,m,
\end{alignat}
where $e_i$ is the unit vector of the $i$\textsuperscript{th} coordinate direction.

Any interior point of any discrete curve $\Sigma$ in $\Z^m$ is of one of the four types shown on Figure \ref{Fig: corners}.\par
\begin{figure}[htbp]
\centering
\subfloat[]{\label{Fig: corners1}
\begin{tikzpicture}[auto,scale=0.6,>=stealth',inner sep=2]
   \node (x-1) at (-2,0) [circle,fill,label=-135:$x_{-i}$] {};
   \node (x) at (0,0) [circle,fill,label=-90:$x$] {};
   \node (x1) at (2,0) [circle,fill,label=-45:$x_i$] {};
   \draw (x) to (x1);
   \draw (x) to (x-1);
\end{tikzpicture}
}\qquad
\subfloat[]{\label{Fig: corners2}
\begin{tikzpicture}[auto,scale=0.6,>=stealth',inner sep=2pt]
   \node (x) at (0,0) [circle,fill,label=-135:$x$] {};
   \node (x1) at (2,0) [circle,fill,label=-45:$x_i$] {};
   \node (x2) at (0,2) [circle,fill,label=135:$x_j$] {};
   \draw (x) to (x1);
   \draw (x) to (x2);
\end{tikzpicture}
}\qquad
\subfloat[]{\label{Fig: corners3}
\begin{tikzpicture}[auto,scale=0.6,>=stealth',inner sep=2pt]
   \node (x) at (0,0) [circle,fill,label=-135:$x_{-i}$] {};
   \node (x1) at (2,0) [circle,fill,label=-45:$x$] {};
   \node (x12) at (2,2) [circle,fill,label=45:$x_{j}$] {};
   \draw (x) to (x1) to (x12);
\end{tikzpicture}
}\qquad
\subfloat[]{\label{Fig: corners4}
\begin{tikzpicture}[auto,scale=0.6,>=stealth',inner sep=2pt]
   \node (x1) at (2,0) [circle,fill,label=-45:$x_{-j}$] {};
   \node (x2) at (0,2) [circle,fill,label=135:$x_{-i}$] {};
   \node (x12) at (2,2) [circle,fill,label=45:$x$] {};
   \draw (x1) to (x12);
   \draw (x2) to (x12);
\end{tikzpicture}
}
\caption{Four type of vertices of a discrete curve. Case \protect\subref{Fig: corners1}: two edges of one coordinate direction meet at $n$. Case \protect\subref{Fig: corners2}: a negatively directed edge followed by a positively directed edge. Case \protect\subref{Fig: corners3}: two equally (positively or negatively) directed edges of two different coordinate directions meet at $n$. Case \protect\subref{Fig: corners4}: a positively directed edge followed by a negatively directed edge.}
\label{Fig: corners}
\end{figure}
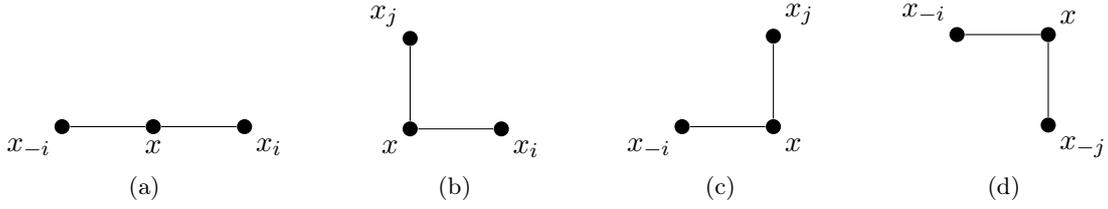
The pieces of discrete curves as on Figures~\ref{Fig: corners}\subref{Fig: corners2}, \subref{Fig: corners3}, and \subref{Fig: corners4} will be called \emph{2D corners}. Observe that a straight piece of a discrete curve, as on Figure \ref{Fig: corners}\subref{Fig: corners1} is a sum of 2D corners, as on Figures~\ref{Fig: corners}\subref{Fig: corners2} and \subref{Fig: corners3}. The whole variety of Euler-Lagrange equations for a pluri-Lagrangian system with $d=1$ reduces to the following three types of \emph{2D corner equations}:
\begin{align}
\label{eq: E 0}
&\frac{\partial\Lambda_i(x,x_i)}{\partial x}-
\frac{\partial\Lambda_j(x,x_j)}{\partial x}=0,
\\
\label{eq: E1 0}
&\frac{\partial\Lambda_i(x_{-i},x)}{\partial x}+
\frac{\partial\Lambda_j(x,x_j)}{\partial x}=0,
\\
\label{eq: E12 0}
&\frac{\partial\Lambda_i(x_{-i},x)}{\partial x}-
\frac{\partial\Lambda_j(x_{-j},x)}{\partial x} = 0.
\end{align}
In particular, the standard single-time discrete Euler-Lagrange equation,
\begin{equation*}\label{eq: dEL}
\frac{\partial\Lambda_i(x_{-i},x)}{\partial x}+\frac{\partial\Lambda_i(x,x_i)}{\partial x}=0,
\end{equation*}
corresponding to a straight piece of a discrete curve as on Figure~\ref{Fig: corners}\subref{Fig: corners1} is a consequence of equations \eqref{eq: E 0} and \eqref{eq: E1 0}, corresponding to 2D corners as on on Figures~\ref{Fig: corners}\subref{Fig: corners2} and \subref{Fig: corners3}.

To discuss the \emph{consistency} of the system of 2D corner equations, it will be more convenient to re-write them with appropriate shifts, as
\begin{align}
\label{eq: E}\tag{$E$}
&\frac{\partial\Lambda_i(x,x_i)}{\partial x}-
\frac{\partial\Lambda_j(x,x_j)}{\partial x}=0,
\\
\label{eq: E1}\tag{$E_i$}
&\frac{\partial\Lambda_i(x,x_i)}{\partial x_i}+
\frac{\partial\Lambda_j(x_i,x_{ij})}{\partial x_i}=0,
\\
\label{eq: E2}\tag{$E_j$}
&\frac{\partial\Lambda_j(x,x_j)}{\partial x_j}+
\frac{\partial\Lambda_i(x_j,x_{ij})}{\partial x_j}=0,
\\
\label{eq: E12}\tag{$E_{ij}$}
&\frac{\partial\Lambda_i(x_j,x_{ij})}{\partial x_{ij}}-
\frac{\partial\Lambda_j(x_i,x_{ij})}{\partial x_{ij}} = 0.
\end{align}
In this form, 2D corner equations \eqref{eq: E}--\eqref{eq: E12} correspond to the four vertices of an elementary square $\sigma_{ij}$ of the lattice, as on Figure~\ref{Fig: consistency}\subref{Fig: consistency1}.
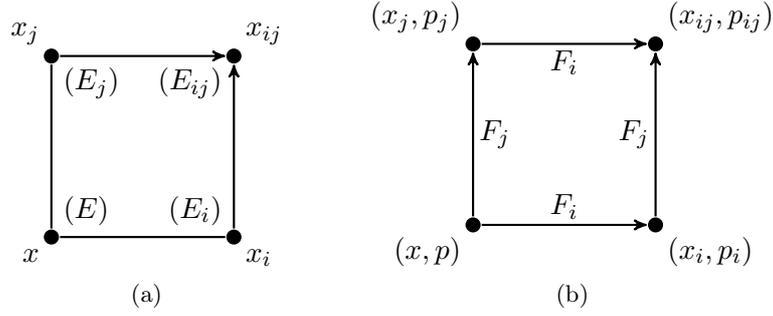
\begin{figure}[htbp]
\centering
\subfloat[]{\label{Fig: consistency1}
\begin{tikzpicture}[auto,scale=1.2,>=stealth',inner sep=2]
   \node (x) at (0,0) [circle,fill,thick,label=-135:$x$,label=45:$\left(E\right)$] {};
   \node (x1) at (2,0) [circle,fill,thick,label=135:$\left(E_{i}\right)$,label=-45:$x_i$] {};
   \node (x2) at (0,2) [circle,fill,thick,label=-45:$\left(E_{j}\right)$,label=135:$x_j$] {};
   \node (x12) at (2,2) [circle,fill,thick,label=45:$x_{ij}$,label=-135:$\left(E_{ij}\right)$] {};
   \draw [thick,->] (x) to (x1) to (x12);
   \draw [thick,->] (x) to (x2) to (x12);
\end{tikzpicture}
}\qquad
\subfloat[]{\label{Fig: consistency2}
\begin{tikzpicture}[auto,scale=1.2,>=stealth',inner sep=2]
   \node (x) at (0,0) [circle,fill,thick,{label=-135:$\left(x,p\right)$}] {};
   \node (x1) at (2,0) [circle,fill,thick,{label=-45:$\left(x_i,p_i\right)$}] {};
   \node (x2) at (0,2) [circle,fill,thick,{label=135:$\left(x_j,p_j\right)$}] {};
   \node (x12) at (2,2) [circle,fill,thick,{label=45:$\left(x_{ij},p_{ij}\right)$}] {};
   \draw [thick,->] (x) to node {$F_{i}$} (x1);
   \draw [thick,->] (x) to node [swap] {$F_{j}$} (x2);
   \draw [thick,->] (x2) to node [swap]{$F_{i}$} (x12);
   \draw [thick,->] (x1) to node {$F_{j}$} (x12);
\end{tikzpicture}
}
\caption{Consistency of 2D corner equations: \protect\subref{Fig: consistency1} Start with data $x$, $x_i$, $x_j$ related by 2D~corner equation \eqref{eq: E}; solve 2D~corner equations \eqref{eq: E1} and \eqref{eq: E2} for $x_{ij}$; consistency means that the two values of $x_{ij}$ coincide identically and satisfy 2D~corner equation \eqref{eq: E12}. \protect\subref{Fig: consistency2} Maps $F_{i}$ and $F_{j}$ commute.}
\label{Fig: consistency}
\end{figure}
Consistency of the system of 2D corner equations \eqref{eq: E}--\eqref{eq: E12} should be understood as follows: start with the fields $x$, $x_i$, $x_j$ satisfying equation \eqref{eq: E}. Then each of equations \eqref{eq: E1}, \eqref{eq: E2} can be solved for $x_{ij}$. Thus, we obtain two alternative values for the latter field. Consistency takes place if these values coincide identically (with respect to the initial data), and, moreover, if the resulting field $x_{ij}$ satisfies equation \eqref{eq: E12}. In other words:
\begin{defi}\label{def: 2D corner eqs consist}
The system of 2D corner equations \eqref{eq: E}--\eqref{eq: E12} is called \emph{consistent}, if it has the minimal possible rank 2, i.e., if exactly two of these four equations are independent.
\end{defi}

Observe that 2D corner equations \eqref{eq: E}--\eqref{eq: E12} can be put as
\begin{alignat}{4}\label{eq: 2D corner eqs}
&\frac{\partial S^{ij}}{\partial x}=0,&\qquad
&\frac{\partial S^{ij}}{\partial x_i}=0,&\qquad
&\frac{\partial S^{ij}}{\partial x_j}=0,&\qquad
&\frac{\partial S^{ij}}{\partial x_{ij}}=0,
\end{alignat}
where $S^{ij}$ is the action along the boundary of an oriented elementary square $\sigma_{ij}$ (this action can be identified with the discrete exterior derivative $d\Ell$ evaluated at $\sigma_{ij}$),
\[
S^{ij}=d\Ell(\sigma_{ij})=\Delta_i\Ell(\sigma_j)-\Delta_j\Ell(\sigma_i)
=\Lambda_i(x,x_i)+\Lambda_j(x_i,x_{ij})-\Lambda_i(x_j,x_{ij})-\Lambda_j(x,x_j).
\]
The main feature of our definition is that the ``almost closedness'' of the 1-form $\Ell$ on solutions of the system of 2D~corner equations is, so to say, built-in from the outset.

\begin{theo}\label{th: discr almost closed}
For any pair of the coordinate directions $i,j$, the action $S^{ij}$ over the boundary of an elementary square of these coordinate directions is constant on solutions of the system of 2D~corner equations \eqref{eq: corner eqs}:
\[
S^{ij}(x,x_i,x_{ij},x_j)=\ell^{ij}=\mathrm{const}  \pmod{\partial S^{ij}/\partial x=0,\ldots,
\partial S^{ij}/\partial x_{ij}=0}.
\]
\end{theo}
In particular, if all these constants $\ell^{ij}$ vanish, then the discrete 1-form $\Ell$ is closed on solutions of the Euler-Lagrange equations, so that the critical value of the action functional $S_\Sigma$ does not depend on the choice of the curve $\Sigma$ connecting two given points in $\Z^m$.
\medskip

We now turn to the Hamiltonian part of the theory of one-dimensional pluri-Lagrangian systems. Consistency of the system of 2D corner equations \eqref{eq: E 0}--\eqref{eq: E12 0} is equivalent to existence of a function $p:\Z^m\to {\mathcal X}$ satisfying all the relations
\begin{alignat}{2}
    p & = \frac{\partial\Lambda_i(x,x_i)}{\partial x},&\qquad& i=1,\ldots,m,
    \label{eq: discr p 1}\\
    p  & = -\frac{\partial\Lambda_i(x_{-i},x)}{\partial x},&\qquad& i=1,\ldots,m.
     \label{eq: discr p 2}
\end{alignat}
We say that the multi-time discrete Lagrangian 1-form $\Ell$ is \emph{Legendre transformable}, if all the equations (\ref{eq: discr p 1}) can be solved for $x_i$ in terms of $x,p$. In this case, equations
\begin{alignat}{2}\label{eq: single Lagr map}
    &p=\frac{\partial\Lambda_i(x,x_i)}{\partial x},&\qquad
    &p_i=-\frac{\partial\Lambda_i(x,x_i)}{\partial x_i},
\end{alignat}
define a symplectic map $F_i:(x,p)\mapsto(x_i,p_i)$.
\begin{theo}
For a consistent one-dimensional pluri-Lagrangian system with a Legendre-trans\-formable 1-form $\Ell$, maps $F_i$ commute:
\begin{equation}\label{eq: commute}
F_i\circ F_j=F_j\circ F_i,
\end{equation}
see Figure~\ref{Fig: consistency}\subref{Fig: consistency2}). Conversely, for a given system of $m$ commuting symplectic maps
$F_i$ admitting Lagrangians (generating functions) $\Lambda_i$, the 1-form $\Ell$ defined by $\Ell(\sigma_i)=\Lambda_i(x,x_i)$, generates a consistent one-dimensional pluri-Lagrangian system.
\end{theo}

\subsection{Two-dimensional pluri-Lagrangian systems, \texorpdfstring{$d=2$}{d=2}}

This section is based on \cite{variational}.

In the case $d=2$, $\Ell$ is a function of oriented elementary squares
\[
\sigma_{ij}=\left(n,n+e_{i},n+e_{i}+e_{j},n+e_{j}\right),
\]
such that $\Ell\left(\sigma_{ij}\right)=-\Ell\left(\sigma_{ji}\right)$.

One can show that the flower of any interior vertex of an oriented quad-surface $\Sigma$ in $\Z^m$ can be represented as a sum of (oriented) 3D~corners in $\Z^{m+1}$. Here, a \emph{3D~corner} is a quad-surface consisting of three elementary squares adjacent to a vertex of valence 3. Examples of 3D corners are given in \cite{variational}.
As a consequence, the action for any flower can be represented as a sum of actions for several 3D~corners. Thus, Euler-Lagrange equation for any interior vertex $n$ of $\Sigma$ can be represented as a sum of several Euler-Lagrange equations for 3D~corners. This justifies the following fundamental definition:
\begin{defi}\label{def:pluriLagr system}
The \emph{system of 3D~corner equations} for a given discrete 2-form $\Ell$ consists of discrete Euler-Lagrange equations for all possible 3D~corners in $\Z^m$. If the action for the surface of an oriented elementary cube $\sigma_{ijk}$ of the coordinate directions $i,j,k$ (which can be identified with the discrete exterior derivative $d\Ell$ evaluated at $\sigma_{ijk}$) is denoted by
\begin{equation}\label{eq: Sijk}
S^{ijk}=d\Ell\left(\sigma_{ijk}\right)=\Delta_k\Ell\left(\sigma_{ij}\right)+
\Delta_i\Ell\left(\sigma_{jk}\right)+\Delta_j\Ell\left(\sigma_{ki}\right),
\end{equation}
then the system of 3D~corner equations consists of the eight equations
\begin{equation}\label{eq: corner eqs}
\begin{alignedat}{4}
&\dfrac{\partial S^{ijk}}{\partial x}=0, &\qquad& \dfrac{\partial S^{ijk}}{\partial x_i}=0, &\qquad& \dfrac{\partial S^{ijk}}{\partial x_j}=0, &\qquad& \dfrac{\partial S^{ijk}}{\partial x_k}=0, \\
\\
&\dfrac{\partial S^{ijk}}{\partial x_{ij}}=0, &\qquad& \dfrac{\partial S^{ijk}}{\partial x_{jk}}=0, &\qquad& \dfrac{\partial S^{ijk}}{\partial x_{ik}}=0, &\qquad& \dfrac{\partial S^{ijk}}{\partial x_{ijk}}=0,
\end{alignedat}
\end{equation}
for each triple $i,j,k$.
\end{defi}
Thus, the system of 3D~corner equations encompasses all possible discrete Euler-Lagrange equations for all possible quad-surfaces $\Sigma$. In other words, solutions of a two-dimensional pluri-Lagrangian problem as introduced in Definition \ref{def:pluriLagr problem} are precisely solutions of the corresponding system of 3D~corner equations.
\begin{rem}
We formulated the system of 3D~corner equations for a generic 2-form $\Ell$. In particular cases the quantity $S^{ijk}$ could be independent on some of the fields at the corners of the cube. Then the system of 3D~corner equations \eqref{eq: corner eqs} could contain less equations.
\end{rem}
Of course, in order that the above definition be meaningful, the system of 3D~corner equations has to be \emph{consistent}:
\begin{defi} \label{def: corner eqs consist}
The system \eqref{eq: corner eqs} is called \emph{consistent}, if it has the minimal possible rank 2, i.e., if exactly two of these equations are independent.
\end{defi}
Again, the ``almost closedness'' of the 2-form $\Ell$ on solutions of the system of 3D~corner equations is built-in from the outset.
\begin{theo}\label{Th: almost closed}
For any triple of the coordinate directions $i,j,k$, the action $S^{ijk}$ over an elementary cube of these coordinate directions is constant on solutions of the system of 3D~corner equations \eqref{eq: corner eqs}:
\[
S^{ijk}\left(x,\ldots,x_{ijk}\right)=c^{ijk}=\mathrm{const}  \pmod{\partial S^{ijk}/\partial x=0,\ldots,\partial S^{ijk}/\partial x_{ijk}=0}.
\]
\end{theo}
The most interesting case is, of course, when all $c^{ijk}=0$. Then $d\Ell=0$, that is, the discrete 2-form $\Ell$ is \emph{closed} on solutions of the system of 3D~corner equations, so that the critical value of the action $S_\Sigma$ does not change under perturbations of the quad-surface $\Sigma$ in $\Z^m$ fixing its boundary.

\section{From 2D pluri-Lagrangian systems to relativistic Toda type systems}
\label{sect: from pluri to rel}

We start with a general 3-point 2-form
\begin{equation}\label{eq: 3point}
    \cL(\sigma_{ij})=L_i(x_i-x)-L_j(x_j-x)-\Lambda_{ij}(x_j-x_i),
\end{equation}
where the Lagrangians $L_i$ and $\Lambda_{ij}$ only depend on the differences of the fields at the end points, and the diagonal Lagrangians are skew-symmetric in the sense that $\Lambda_{ij}(x)=-\Lambda_{ji}(-x)$.

For a 3-point 2-form, expression \eqref{eq: Sijk} specializes to
\begin{equation}\label{eq: 3point S}
\begin{split}
S^{ijk} & =  L_i(x_{ik}-x_k)+L_j(x_{ij}-x_i)+L_k(x_{jk}-x_j)\\
        &   \phantom{=}\ -L_i(x_{ij}-x_j)-L_j(x_{jk}-x_k)-L_k(x_{ik}-x_i)\\
        &   \phantom{=}\ -\Lambda_{ij}(x_{jk}-x_{ik})-\Lambda_{jk}(x_{ik}-x_{ij})-\Lambda_{ki}(x_{ij}-x_{jk})\\
        &   \phantom{=}\ +\Lambda_{ij}(x_j-x_i)+\Lambda_{jk}(x_k-x_j)+\Lambda_{ki}(x_i-x_k).
\end{split}
\end{equation}
Thus, $S^{ijk}$ depends on neither $x$ nor $x_{ijk}$, and its domain of definition is better visualized as an octahedron shown in Figure~\ref{octahedron}.

\begin{figure}[htbp]
   \centering
   \begin{tikzpicture}[scale=0.85,inner sep=2]
      \node (x) at (0,0) [circle,fill,label=-135:$x$] {};
      \node (x1) at (3,0) [circle,fill,label=-45:$x_{i}$] {};
      \node (x2) at (1,1) [circle,fill,label=-90:$x_{j}$] {};
      \node (x3) at (0,3) [circle,fill,label=135:$x_{k}$] {};
      \node (x12) at (4,1) [circle,fill,label=0:$x_{ij}$] {};
      \node (x13) at (3,3) [circle,fill,label=0:$x_{ik}$] {};
      \node (x23) at (1,4) [circle,fill,label=135:$x_{jk}$] {};
      \node (x123) at (4,4) [circle,fill,label=45:$x_{ijk}$] {};
      \draw (x) to (x1);
      \draw (x) to (x2);
      \draw (x) to (x3);
      \draw [ultra thick] (x1) to (x2);
      \draw [ultra thick] (x1) to (x3);
      \draw [ultra thick] (x1) to (x12);
      \draw [ultra thick] (x1) to (x13);
      \draw [ultra thick] (x2) to (x3);
      \draw [ultra thick,dashed] (x2) to (x12);
      \draw [ultra thick,dashed] (x2) to (x23);
      \draw [ultra thick] (x3) to (x13);
      \draw [ultra thick] (x3) to (x23);
      \draw [ultra thick] (x12) to (x13);
      \draw (x12) to (x123);
      \draw [ultra thick,dashed] (x12) to (x23);
      \draw [ultra thick] (x13) to (x23);
      \draw (x13) to (x123);
      \draw (x23) to (x123);
   \end{tikzpicture}
   \caption{Octahedron supporting $d\Ell$ for a 3-point discrete 2-form $\Ell$}
   \label{octahedron}
\end{figure}
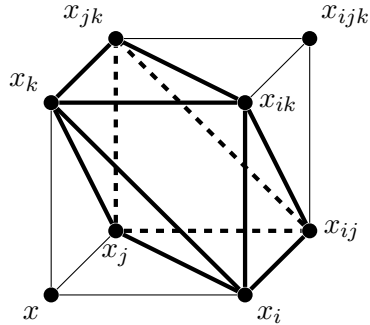

Accordingly, the system of corner equations consists of six equations per elementary 3D cube, which we denote by $({\mathcal E}_i)$, $({\mathcal E}_j)$, $({\mathcal E}_k)$, $({\mathcal E}_{ij})$, $({\mathcal E}_{ik})$, and $({\mathcal E}_{jk})$. Our main assumption is that \emph{the system of corner equations is consistent}. To write them down, we set
\begin{alignat}{2}\label{eq: phi psi}
    &\psi_i(x)=\frac{\partial L_i(x)}{\partial x},&\qquad
    &\phi_{ij}(x)=\frac{\partial \Lambda_{ij}(x)}{\partial x}.
\end{alignat}
In particular, we have: $\phi_{ij}(x)=\phi_{ji}(-x)$. In terms of these functions, corner equations read:
\begin{align}\label{eq: 3point Ei}\tag{${\mathcal E}_i$}
&\psi_j(x_{ij}-x_i)+\phi_{ij}(x_j-x_i)=\psi_k(x_{ik}-x_i)+\phi_{ik}(x_k-x_i),\\
\label{eq: 3point Eij}\tag{${\mathcal E}_{ij}$}
&\psi_j(x_{ij}-x_i)+\phi_{kj}(x_{ij}-x_{ik})=\psi_i(x_{ij}-x_j)+\phi_{ki}(x_{ij}-x_{jk}).
\end{align}
In what follows, one of the coordinate directions (which we denote as the 0\textsuperscript{th} one) plays a distinguished role, it enumerates the sites of the relativistic Toda chains. We will use the index $n$ for this coordinate direction only. Accordingly, we will only consider surfaces in $\mathbb Z^m$ which contain, along with any point, the whole line through this point parallel to the 0\textsuperscript{th} coordinate axis. One can call such surfaces \emph{cylindrical}. The set of values of $x$ along such a line, $x=\{x_n: n\in\mathbb Z\}$, or, upon a finite-dimensional reduction, $x=\{x_n: 1\le n\le N\}$, is an element of the configuration space of the relativistic Toda lattice. We use the accents $\widetilde{\phantom{x}}$ and $\widehat{\phantom{x}}$ to denote the shift in the discrete times corresponding to all other coordinate directions.

\begin{defi}\label{def: Fi}
The map $F_i: (x,p)\mapsto(\wx,\wip)$ is the symplectic map with the generating function
\begin{equation}\label{eq: dRTL1 Lagr}
\ELL_i(x,\wx)=\sum_{n=1}^N L_i(\wx_n-x_n)-\sum_{n=1}^N L_0(x_{n+1}-x_n)
-\sum_{n=1}^N\Lambda_{i0}(x_{n+1}-\wx_n),
\end{equation}
thus its equations of motion $p_n=-\partial \ELL_{i}/\partial x_n$, $\wip_n=\partial \ELL_{i}/\partial \wx_n$ read:
\begin{equation}
F_i:\begin{cases}
p_n=\psi_i(\wx_n-x_n)+\phi_{i0}(x_n-\wx_{n-1})-\psi_0(x_{n+1}-x_n)+\psi_0(x_n-x_{n-1}),\\
\wip_n=\psi_i(\wx_n-x_n)+\phi_{i0}(x_{n+1}-\wx_n).\end{cases}
\end{equation}
The Euler-Lagrange equations read
\begin{equation}
\psi_i(\wx_n-x_n)-\psi_i(x_n-\undertilde{x}_n)=
\psi_0(x_{n+1}-x_n)-\psi_0(x_n-x_{n-1})+\phi_{i0}(\undertilde{x}_{n+1}-x_n)-\phi_{i0}(x_n-\wx_{n-1}).
\end{equation}
\end{defi}

This map corresponds to the edges $(x,\wx)=(x,x_i)$ of the $i$\textsuperscript{th} coordinate direction, to which the strip supporting $\ELL_{i}$ projects along the 0\textsuperscript{th} coordinate axis. See the identifications of variables on Figure~\ref{fig: action for F}. We denote the index set of the maps $F_i$ by $I=\{i,j,\ldots\}$.
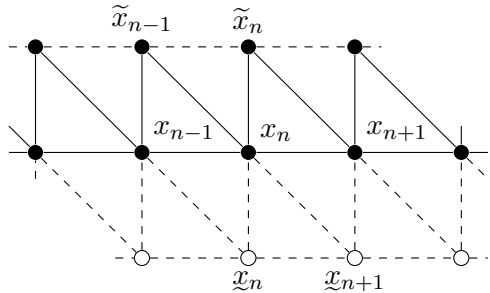
\begin{figure}[htbp]
   \centering
   \begin{tikzpicture}[scale=0.7,inner sep=2]
      \node (x1) at (2,0) [circle,draw] {};
      \node (x11) at (4,0) [circle,draw,label=-90:$\undertilde{x}_{n}$] {};
      \node (x111) at (6,0) [circle,draw,label=-90:$\undertilde{x}_{n+1}$] {};
      \node (x1111) at (8,0) [circle,draw] {};
      \node (x2) at (0,2) [circle,fill] {};
      \node (x12) at (2,2) [circle,fill,label=45:$x_{n-1}$] {};
      \node (x112) at (4,2) [circle,fill,label=45:$x_{n}$] {};
      \node (x1112) at (6,2) [circle,fill,label=45:$x_{n+1}$] {};
      \node (x11112) at (8,2) [circle,fill] {};
      \node (x22) at (0,4) [circle,fill] {};
      \node (x122) at (2,4) [circle,fill,label=90:$\widetilde{x}_{n-1}$] {};
      \node (x1122) at (4,4) [circle,fill,label=90:$\widetilde{x}_{n}$] {};
      \node (x11122) at (6,4) [circle,fill] {};
      \draw [dashed] (1.5,0) to (x1) to (x11) to (x111) to (x1111) to (8.5,0);
      \draw [dashed] (0,1.5) to (x2) to (x1) to (x12) to (x11) to (x112) to (x111) to (x1112) to (x1111) to (x11112) to (8.5,1.5);
      \draw (-0.5,2) to (x2) to (x12) to (x112) to (x1112) to (x11112) to (8.5,2);
      \draw (-0.5,2.5) to (x2) to (x22) to (x12) to (x122) to (x112) to (x1122) to (x1112) to (x11122) to (x11112) to (8,2.5);
      \draw [dashed] (-0.5,4) to (x22) to (x122) to (x1122) to (x11122) to (6.5,4);
   \end{tikzpicture}
   \caption{Domain of the map $F_{i}$}
   \label{fig: action for F}
\end{figure}

\begin{defi}  \label{def: Gk}
The map $G_k: (x,p)\mapsto(\wx,\wip)$ is the symplectic map with the generating function
\begin{equation}\label{eq: dRTL2 Lagr}
\M_k(x,\wx)=\sum_{n=1}^N \Lambda_{k0}(\wx_n-x_n)+\sum_{n=1}^N L_0(\wx_n-\wx_{n-1})-
\sum_{n=1}^N L_k(x_n-\wx_{n-1}),
\end{equation}
thus its equations of motion $p_n=-\partial \M_{k}/\partial x_n$, $\wip_n=\partial \M_{k}/\partial \wx_n$ read:
\begin{equation}
G_k:\begin{cases}
p_n=\phi_{k0}(\wx_n-x_n)+\psi_k(x_n-\wx_{n-1}),\\
\wip_n=\phi_{k0}(\wx_n-x_n)+\psi_k(x_{n+1}-\wx_n)-\psi_0(\wx_{n+1}-\wx_n)+\psi_0(\wx_n-\wx_{n-1}).\end{cases}
\end{equation}
The Euler-Lagrange equations read
\begin{equation}
\phi_{k0}(\wx_n-x_n)-\phi_{k0}(x_n-\undertilde{x}_n)=
-\psi_0(x_{n+1}-x_n)+\psi_0(x_n-x_{n-1})+\psi_k(\undertilde{x}_{n+1}-x_n)-\psi_j(x_n-\wx_{n-1}).
\end{equation}
\end{defi}
This map corresponds to the negatively directed edges $(x,\wx)=(x,x_{-k})$ of the $k$\textsuperscript{th} coordinate direction, to which the strip supporting $\M_{k}$ projects along the 0\textsuperscript{th} coordinate axis. See the identifications of variables on Figure~\ref{fig: action for G}. We denote the index set of the maps $G_k$ by $K=\{k,\ell,\ldots\}$, and assume it to be disjoint from $I=\{i,j,\ldots\}$.
\begin{figure}[htbp]
   \centering
   \begin{tikzpicture}[scale=0.7,inner sep=2]
      \node (x1) at (2,0) [circle,fill] {};
      \node (x11) at (4,0) [circle,fill,label=-90:$\widetilde{x}_{n-1}$] {};
      \node (x111) at (6,0) [circle,fill,label=-90:$\widetilde{x}_{n}$] {};
      \node (x1111) at (8,0) [circle,fill] {};
      \node (x2) at (0,2) [circle,fill] {};
      \node (x12) at (2,2) [circle,fill,label=45:$x_{n-1}$] {};
      \node (x112) at (4,2) [circle,fill,label=45:$x_{n}$] {};
      \node (x1112) at (6,2) [circle,fill,label=45:$x_{n+1}$] {};
      \node (x11112) at (8,2) [circle,fill] {};
      \node (x22) at (0,4) [circle,draw] {};
      \node (x122) at (2,4) [circle,draw,label=90:$\undertilde{x}_{n}$] {};
      \node (x1122) at (4,4) [circle,draw,label=90:$\undertilde{x}_{n+1}$] {};
      \node (x11122) at (6,4) [circle,draw] {};
      \draw (1.5,0) to (x1) to (x11) to (x111) to (x1111) to (8.5,0);
      \draw (0,1.5) to (x2) to (x1) to (x12) to (x11) to (x112) to (x111) to (x1112) to (x1111) to (x11112) to (8.5,1.5);
      \draw [dashed] (-0.5,2) to (x2) to (x12) to (x112) to (x1112) to (x11112) to (8.5,2);
      \draw [dashed] (-0.5,2.5) to (x2) to (x22) to (x12) to (x122) to (x112) to (x1122) to (x1112) to (x11122) to (x11112) to (8,2.5);
      \draw [dashed] (-0.5,4) to (x22) to (x122) to (x1122) to (x11122) to (6.5,4);
   \end{tikzpicture}
   \caption{Domain of the map $G_{k}$}
   \label{fig: action for G}
\end{figure}
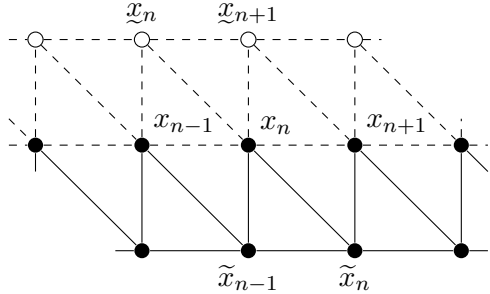

In the present paper, we consider maps $F_{i}$, $G_{k}$ with finitely many degrees of freedom $(1\le n\le N)$. This requires to specify certain boundary conditions. We will consider either the so-called open-end or periodic boundary conditions.
\begin{itemize}
\item Open-end boundary conditions correspond to letting the second and the third sums in the Lagrangian functions \eqref{eq: dRTL1 Lagr}, \eqref{eq: dRTL2 Lagr} extend over $1\leq n\leq N-1$ only. Effectively, this amounts to omitting terms containing $x_{0}$ or $\tilde{x}_{0}$ from the expressions for $p_{1}$ and $\widetilde{p}_{1}$, and likewise omitting terms containing $x_{N+1}$ or $\tilde{x}_{N+1}$ from the expressions for $p_{N}$ and $\widetilde{p}_{N}$. In this case, maps $F_{i}$ and $G_{k}$ are \emph{single-valued} functions of $\left(x,p\right)$.
\item Periodic boundary conditions correspond to letting all indices be taken $\mathrm{mod}\ N$, so that $x_0=x_N$, $x_{N+1}=x_1$. In this case, these maps are \emph{double-valued}, so that the very notion of their commutativity has to be clarified. We achieve this along the same lines as in the previous work \cite{nonrel}.
\end{itemize}

\begin{theo}
Let $i$, $j$, $k$, and $\ell$ be four different indices from $I=\{i,j,\ldots\}$ and $K=\{k,\ell,\ldots\}$. Then any two of the maps $F_i$, $F_j$, $G_k$, and $G_\ell$ commute.
\end{theo}
The next three subsections are devoted to the proof of this theorem. We prove separately the commutativity of $F_i$ and $F_j$, of $G_k$ and $G_\ell$, and of $F_i$ and $G_\ell$. As explained in Section~\ref{one-dimensional}, each such statement is equivalent to consistency of the corresponding system of 2D corner equations.

\subsection{Proof of commutativity of the maps \texorpdfstring{$F_i$, $F_j$}{F\_i, F\_j}}
\label{sect: rtl++}

The 2D~corner equations for the pluri-Lagrangian system corresponding to two maps $F_{i}$ and $F_{j}$ read:
\begin{align}
\label{eq: rtl E}\tag{$E$}
&\begin{aligned}\psi_i(\wx_n-x_n)+\phi_{i0}(x_n-\wx_{n-1})=\psi_j(\whx_n-x_n)+\phi_{j0}(x_n-\whx_{n-1}),
\end{aligned}\\\label{eq: rtl E1}\tag{$E_i$}
&\begin{aligned}
&\psi_i(\wx_n-x_n)+\phi_{i0}(x_{n+1}-\wx_n)=
\psi_j(\widehat{\wx}_n-\wx_n)+\phi_{j0}(\wx_n-\widehat{\wx}_{n-1})\\
&\hspace{7cm}-\psi_0(\wx_{n+1}-\wx_n)+\psi_0(\wx_n-\wx_{n-1}),&
\end{aligned}\\
\label{eq: rtl E2}\tag{$E_j$}
&\begin{aligned}
&\psi_j(\whx_n-x_n)+\phi_{j0}(x_{n+1}-\whx_n)=
\psi_i(\widehat{\wx}_n-\whx_n)+\phi_{i0}(\whx_n-\widehat{\wx}_{n-1})\\
&\hspace{7cm}-\psi_0(\whx_{n+1}-\whx_n)+\psi_0(\whx_n-\whx_{n-1}),&
\end{aligned}\\
\label{eq: rtl E12}\tag{$E_{ij}$}
&\begin{aligned} \psi_i(\widehat{\wx}_n-\whx_n)+\phi_{i0}(\whx_{n+1}-\widehat{\wx}_n)=
\psi_j(\widehat{\wx}_n-\wx_n)+\phi_{j0}(\wx_{n+1}-\widehat{\wx}_n).
\end{aligned}
\end{align}
\begin{figure}[htb]
   \centering
   \subfloat[Domain of the equation \eqref{eq: rtl E}]{\label{fig:6a}
   \begin{tikzpicture}[scale=0.6,inner sep=2]
      \node (x) at (0,0) [circle,fill,label=-90:$x_{n-1}$] {};
      \node (x1) at (3,0) [circle,fill,label=-90:$x_{n}$] {};
      \node (x2) at (1,1) [circle,fill,label=135:$\widehat{x}_{n-1}$] {};
      \node (x3) at (0,3) [circle,fill,label=90:$\widetilde{x}_{n-1}$] {};
      \node (x11) at (6,0) [circle,fill,label=-90:$x_{n+1}$] {};
      \node (x12) at (4,1) [circle,fill,label=90:$\widehat{x}_{n}$] {};
      \node (x13) at (3,3) [circle,fill,label=90:$\widetilde{x}_{n}$] {};
      \node (x112) at (7,1) [circle,fill,label=90:$\widehat{x}_{n+1}$] {};
      \node (x113) at (6,3) [circle,fill,label=90:$\widetilde{x}_{n+1}$] {};
      \node (y) at (9,0) [circle,fill,label=-90:$x$] {};
      \node (y2) at (10,1) [circle,fill,label=90:$\widehat{x}$] {};
      \node (y3) at (9,3) [circle,fill,label=90:$\widetilde{x}$] {};
      \draw (x) to (x1);
      \draw (x) to (x2);
      \draw (x) to (x3);
      \draw [ultra thick] (x1) to (x2);
      \draw [ultra thick] (x1) to (x3);
      \draw (x1) to (x11);
      \draw [ultra thick] (x1) to (x12);
      \draw [ultra thick] (x1) to (x13);
      \draw (x2) to (x12);
      \draw (x3) to (x13);
      \draw (x11) to (x112);
      \draw (x11) to (x113);
      \draw (x12) to (x112);
      \draw (x13) to (x113);
      \draw [ultra thick] (y2) to (y) to (y3);
      \node (x2) at (1,1) [circle,fill,label={135,fill=white}:$\widehat{x}_{n-1}$] {};
   \end{tikzpicture}
   }\qquad
   \subfloat[Domain of the equation \eqref{eq: rtl E1}]{\label{fig:6b}
   \begin{tikzpicture}[scale=0.6,inner sep=2]
      \node (x) at (0,0) [circle,fill,label=-90:$x_{n-1}$] {};
      \node (x1) at (3,0) [circle,fill,label=-90:$x_{n}$] {};
      \node (x3) at (0,3) [circle,fill,label=-135:$\widetilde{x}_{n-1}$] {};
      \node (x11) at (6,0) [circle,fill,label=-90:$x_{n+1}$] {};
      \node (x13) at (3,3) [circle,fill,label=-135:$\widetilde{x}_{n}$] {};
      \node (x23) at (1,4) [circle,fill,label=90:$\htilde{x}_{n-1}$] {};
      \node (x113) at (6,3) [circle,fill,label=-45:$\widetilde{x}_{n+1}$] {};
      \node (x123) at (4,4) [circle,fill,label=90:$\htilde{x}_{n}$] {};
      \node (x1123) at (7,4) [circle,fill,label=90:$\htilde{x}_{n+1}$] {};
      \node (y) at (9,0) [circle,fill,label=-90:$x$] {};
      \node (y23) at (10,4) [circle,fill,label=90:$\htilde{x}$] {};
      \node (y3) at (9,3) [circle,fill,label=-45:$\widetilde{x}$] {};
      \draw (x) to (x1);
      \draw (x) to (x1);
      \draw (x) to (x3);
      \draw (x1) to (x11);
      \draw [ultra thick] (x1) to (x13);
      \draw [ultra thick] (x3) to (x13);
      \draw (x3) to (x23);
      \draw [ultra thick] (x11) to (x13);
      \draw (x11) to (x113);
      \draw [ultra thick] (x13) to (x23);
      \draw [ultra thick] (x13) to (x113);
      \draw [ultra thick] (x13) to (x123);
      \draw (x23) to (x123);
      \draw (x113) to (x1123);
      \draw (x123) to (x1123);
      \draw [ultra thick] (y23) to (y3) to (y);
   \end{tikzpicture}
   }\\
   \subfloat[Domain of the equation \eqref{eq: rtl E2}]{\label{fig:6c}
   \begin{tikzpicture}[scale=0.6,inner sep=2]
      \node (x) at (0,0) [circle,fill,label=-90:$x_{n-1}$] {};
      \node (x1) at (3,0) [circle,fill,label=-90:$x_{n}$] {};
      \node (x2) at (1,1) [circle,fill,label=135:$\widehat{x}_{n-1}$] {};
      \node (x11) at (6,0) [circle,fill,label=-90:$x_{n+1}$] {};
      \node (x12) at (4,1) [circle,fill,label=45:$\widehat{x}_{n}$] {};
      \node (x23) at (1,4) [circle,fill,label=90:$\htilde{x}_{n-1}$] {};
      \node (x112) at (7,1) [circle,fill,label=45:$\widehat{x}_{n+1}$] {};
      \node (x123) at (4,4) [circle,fill,label=90:$\htilde{x}_{n}$] {};
      \node (x1123) at (7,4) [circle,fill,label=90:$\htilde{x}_{n+1}$] {};
      \node (y) at (9,0) [circle,fill,label=-90:$x$] {};
      \node (y2) at (10,1) [circle,fill,label=45:$\widehat{x}$] {};
      \node (y23) at (10,4) [circle,fill,label=90:$\htilde{x}$] {};
      \draw (x) to (x1);
      \draw (x) to (x2);
      \draw (x1) to (x11);
      \draw [ultra thick] (x1) to (x12);
      \draw [ultra thick] (x2) to (x12);
      \draw (x2) to (x23);
      \draw [ultra thick] (x11) to (x12);
      \draw (x11) to (x112);
      \draw [ultra thick] (x12) to (x23);
      \draw [ultra thick] (x12) to (x112);
      \draw [ultra thick] (x12) to (x123);
      \draw (x23) to (x123);
      \draw (x112) to (x1123);
      \draw (x123) to (x1123);
      \draw [ultra thick] (y23) to (y2) to (y);
   \end{tikzpicture}
   }\qquad
      \subfloat[Domain of the equation \eqref{eq: rtl E12}]{\label{fig:6d}
   \begin{tikzpicture}[scale=0.6,inner sep=2]
      \node (x2) at (1,1) [circle,fill,label=-90:$\widehat{x}_{n-1}$] {};
      \node (x3) at (0,3) [circle,fill,label=-90:$\widetilde{x}_{n-1}$] {};
      \node (x12) at (4,1) [circle,fill,label=-90:$\widehat{x}_{n}$] {};
      \node (x13) at (3,3) [circle,fill,label=-90:$\widetilde{x}_{n}$] {};
      \node (x23) at (1,4) [circle,fill,label=90:$\htilde{x}_{n-1}$] {};
      \node (x112) at (7,1) [circle,fill,label=-90:$\widehat{x}_{n+1}$] {};
      \node (x113) at (6,3) [circle,fill,label=-45:$\widetilde{x}_{n+1}$] {};
      \node (x123) at (4,4) [circle,fill,label=90:$\htilde{x}_{n}$] {};
      \node (x1123) at (7,4) [circle,fill,label=90:$\htilde{x}_{n+1}$] {};
      \node (y23) at (10,4) [circle,fill,label=90:$\htilde{x}$] {};
      \node (y2) at (10,1) [circle,fill,label=-90:$\widehat{x}$] {};
      \node (y3) at (9,3) [circle,fill,label=-90:$\widetilde{x}$] {};
      \draw (x2) to (x12);
      \draw (x2) to (x23);
      \draw (x3) to (x13);
      \draw (x3) to (x23);
      \draw (x12) to (x112);
      \draw [ultra thick] (x12) to (x123);
      \draw (x13) to (x113);
      \draw [ultra thick] (x13) to (x123);
      \draw (x23) to (x123);
      \draw [ultra thick] (x112) to (x123);
      \draw (x112) to (x1123);
      \draw [ultra thick] (x113) to (x123);
      \draw (x113) to (x1123);
      \draw (x123) to (x1123);
      \draw [ultra thick] (y2) to (y23) to (y3);
      \node (x113) at (6,3) [circle,fill,label={-45,fill=white}:$\widetilde{x}_{k+1}$] {};
   \end{tikzpicture}
   }
   \caption{2D~corner equations for the system of $F_i$ and $F_j$: equations $(E)$ and $(E_{ij})$ are 3D~corner equations of the corresponding discrete 2-form, while equations $(E_i)$ and $(E_j)$ are sums of two 3D~corner equations coming from two 3D~corners with one common face.}
   \label{fig: 2d for Fi, Fj}
\end{figure}
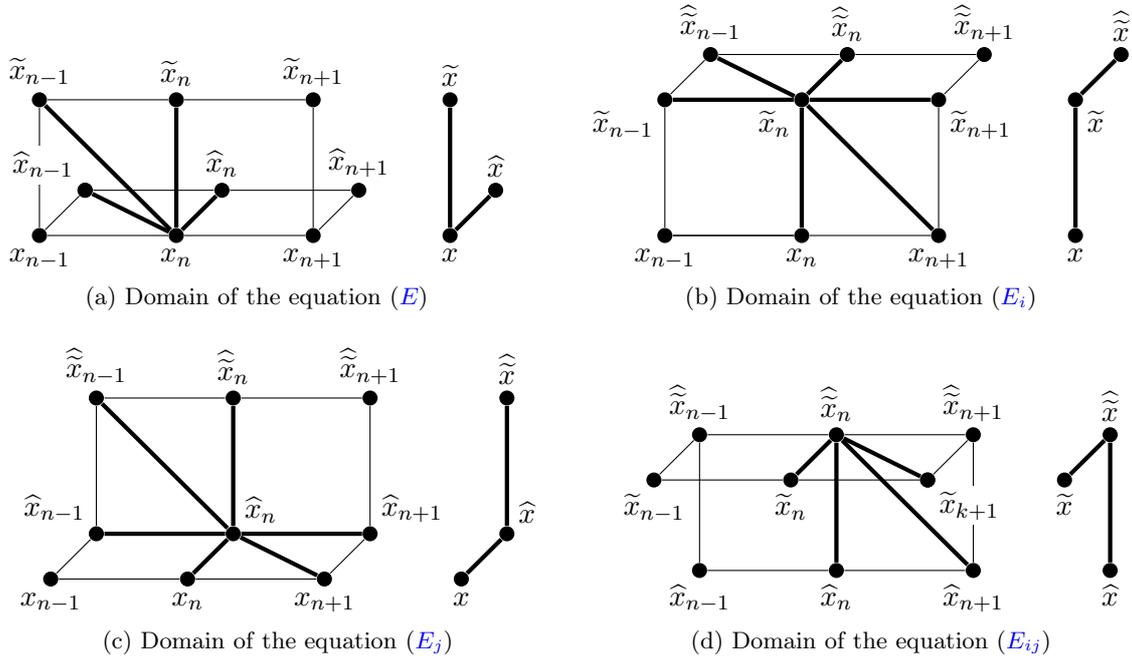
A visualization of the 2D~corner equations embedded in $\Z^{3}$ is given in Figure~\ref{fig: 2d for Fi, Fj}. Discrete curves in the multi-time plane $\Z^{2}$ are in a one-to-one correspondence with cylindrical surfaces in $\Z^3$, via the projection along the first coordinate direction of $\Z^3$. For 2D corners, this is illustrated in Figure~\ref{fig: 2d for Fi, Fj}.

Consistency of the above system of 2D~corner equations is proven with the help of the following statement.
\begin{theo}\label{th: rtl+}
Suppose that the fields $x$, $\wx$, and $\whx$ satisfy 2D~corner equations \eqref{eq: rtl E}. Define the fields $\widehat{\wx}$ by any of the following four formulae, which are equivalent by virtue of \eqref{eq: rtl E}:
\begin{align}
\label{eq: rtl S1}\tag{$S1$}
& \psi_j(\widehat{\wx}_n-\wx_n)+\phi_{ij}(\whx_n-\wx_n)=
\psi_0(\wx_{n+1}-\wx_n)+\phi_{i0}(x_{n+1}-\wx_n), \\
\label{eq: rtl S2}\tag{$S2$}
& \psi_i(\widehat{\wx}_n-\whx_n)+\phi_{ji}(\wx_n-\whx_n)=
\psi_0(\whx_{n+1}-\whx_n)+\phi_{j0}(x_{n+1}-\whx_n), \\
\label{eq: rtl S3}\tag{$S3$}
& \psi_i(\wx_{n+1}-x_{n+1})+\phi_{ji}(\wx_{n+1}-\whx_{n+1})=
\psi_0(\wx_{n+1}-\wx_n)+\phi_{j0}(\wx_{n+1}-\widehat{\wx}_n), \\
\label{eq: rtl S4}\tag{$S4$}
& \psi_j(\whx_{n+1}-x_{n+1})+\phi_{ij}(\whx_{n+1}-\wx_{n+1})=
\psi_0(\whx_{n+1}-\whx_n)+\phi_{i0}(\whx_{n+1}-\widehat{\wx}_n),
\end{align}
called superposition formulae (note that each one of these formulae is local with respect to $\widehat{\wx}$). Then the 2D~corner equations \eqref{eq: rtl E1}, \eqref{eq: rtl E2}, and \eqref{eq: rtl E12} are satisfied, as well.
\end{theo}
\begin{proof}
Consider the sublattice $\Z^3$ spanned by the coordinate directions 0 (indexed by the the letter $n$), $i$, corresponding to the map $F_i$ (shift in this direction being denoted by $\widetilde{\hspace{5pt}}$\ ), and $j$, corresponding to the map $F_{\mu}$ (shift in this direction being denoted by $\widehat{\hspace{5pt}}$\ ). See Figure \ref{fig: 3d for Fi, Fj}.
\begin{figure}[htbp]
   \centering
   \begin{tikzpicture}[auto,scale=0.7,inner sep=2,>=stealth']
      \node (x) at (0,0) [circle,fill,label=-135:{$x_{n}\simeq x$}] {};
      \node (x1) at (3,0) [circle,fill,label=-45:{$x_{n+1}\simeq x_0$}] {};
      \node (x2) at (1,1) [circle,fill,label=135:{$\whx_{n}\simeq x_j$}] {};
      \node (x3) at (0,3) [circle,fill,label=-135:{$\wx_{n}\simeq x_i$}] {};
      \node (x12) at (4,1) [circle,fill,label=45:{$\whx_{n+1}\simeq x_{j0}$}] {};
      \node (x23) at (1,4) [circle,fill,label=135:{$\widehat{\wx}_{n}\simeq x_{ij}$}] {};
      \node (x13) at (3,3) [circle,fill,label=-45:{$\wx_{n+1}\simeq x_{i0}$}] {};
      \node (x123) at (4,4) [circle,fill,label=45:{$\widehat{\wx}_{n+1}\simeq x_{ij0}$}] {};
      \draw (x) to (x1) to (x13) to (x3) to (x);
      \draw (x1) to (x12) to (x123) to (x23) to (x3);
      \draw (x13) to (x123);
      \draw [dashed] (x) to (x2) to (x12);
      \draw [dashed] (x2) to (x23);
      \node (x13) at (3,3) [circle,fill,label={-45,fill=white:{$\wx_{n+1}\simeq x_{i0}$}}] {};
      \node (x2) at (1,1) [circle,fill,label={135,fill=white:{$\whx_{n}\simeq x_j$}}] {};
      \draw [ultra thick,<->] (10,1) to node {$F_{j}$} (9,0) to node {$F_{i}$} (9,3);
   \end{tikzpicture}
   \caption{Identification of fields: the $0$\textsuperscript{th} coordinate direction with the space-direction, the $i$\textsuperscript{th} coordinate direction with the time-direction of the map $F_{i}$ and the $j$\textsuperscript{th} direction with the time-direction of the map $F_{j}$.}
   \label{fig: 3d for Fi, Fj}
\end{figure}
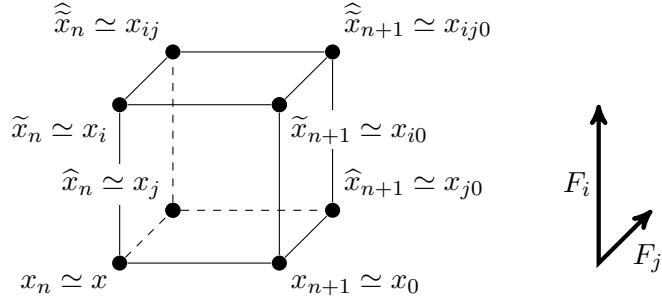

One easily checks that, upon the identifications as on Figure~\ref{fig: 3d for Fi, Fj}, the two corner equations \eqref{eq: rtl E}, \eqref{eq: rtl E12} and the four superposition formulae \eqref{eq: rtl S1}--\eqref{eq: rtl S4} build nothing but the system of 3D~corner equations \eqref{eq: 3point Ei}, \eqref{eq: 3point Eij}. Due to consistency of the latter system, as formulated in Theorem \ref{th: rtl+}, if equation \eqref{eq: rtl E} and one of equations \eqref{eq: rtl S1}--\eqref{eq: rtl S4} hold, then  equation \eqref{eq: rtl E12} and the remaining three of equations \eqref{eq: rtl S1}--\eqref{eq: rtl S4} are satisfied, as well. Furthermore, equation \eqref{eq: rtl E1} is the difference of \eqref{eq: rtl S1} and the downshifted version of \eqref{eq: rtl S3}, while equation \eqref{eq: rtl E2} is the difference of \eqref{eq: rtl S2} and the downshifted version of \eqref{eq: rtl S4}. This completes the proof.
\end{proof}

This theorem provides us with an exhaustive understanding of commutativity of double-valued B\"acklund transformations in the periodic case:
\begin{itemize}
\item In the Lagrangian picture, suppose that we are given fields $x$, $\wx$, $\whx$ satisfying the 2D~corner equation \eqref{eq: rtl E}. Each of equations \eqref{eq: rtl E1}, \eqref{eq: rtl E2} produces two values for $\widehat{\wx}$. Consistency is reflected in the following fact: one of the values for  $\widehat{\wx}$ obtained from \eqref{eq: rtl E1} coincides with one of the values for  $\widehat{\wx}$ obtained from \eqref{eq: rtl E2}. Indeed, this common value is nothing but  $\widehat{\wx}$ obtained from the superposition formulae \eqref{eq: rtl S1}, \eqref{eq: rtl S2}, \eqref{eq: rtl S3} or \eqref{eq: rtl S4}, as in Theorem~\ref{th: rtl+}.
\item In the symplectic maps picture, each of the compositions $F_i\circ F_j$ and $F_j\circ F_i$ applied to a point $(x,p)$ produces four different branches for $(\widehat{\wx},\widehat{\wip})$. Commutativity is reflected in the following fact: each of the branches of $F_i\circ F_j$ coincides with one of the branches of $F_j\circ F_i$. Indeed, Theorem~\ref{th: rtl+} delivers four possible values for $(\wx,\whx,\widehat{\wx})$ satisfying all 2D~corner equations \eqref{eq: rtl E}--\eqref{eq: rtl E12}, namely one $\widehat{\wx}$ for each of the four possible combinations of $\left(\wx,\whx\right)$.
\end{itemize}
The reader is referred to \cite{nonrel} for a graphical illustration and more details.

\subsection{Proof of commutativity of the maps \texorpdfstring{$G_k$, $G_\ell$}{G\_k, G\_l}}
\label{sect: rtl--}

The 2D~corner equations for the pluri-Lagrangian system corresponding to the two maps $G_k$ and $G_\ell$ read:
\begin{align}
\label{eq: rtl- E}\tag{$E$}
&\begin{aligned} \phi_{k0}(\wx_n-x_n)+\psi_k(x_n-\wx_{n-1})=
  \phi_{\ell 0}(\whx_n-x_n)+\psi_\ell(x_n-\whx_{n-1}),
\end{aligned}\\\label{eq: rtl- E1}\tag{$E_k$}
&\begin{aligned}
&\phi_{k0}(\wx_n-x_n)+\psi_k(x_{n+1}-\wx_n)-\psi_0(\wx_{n+1}-\wx_n)+\psi_0(\wx_n-\wx_{n-1})\\
&\hspace{7cm}=\phi_{\ell 0}(\widehat{\wx}_n-\wx_n)+\psi_\ell(\wx_n-\widehat{\wx}_{n-1}),
\end{aligned}\\\label{eq: rtl- E2}\tag{$E_\ell$}
&\begin{aligned}
&\phi_{\ell 0}(\whx_n-x_n)+\psi_\ell(x_{n+1}-\whx_n)-\psi_0(\whx_{n+1}-\whx_n)+\psi_0(\whx_n-\whx_{n-1})\\
&\hspace{7cm}=\phi_{k0}(\widehat{\wx}_n-\whx_n)+\psi_k(\whx_n-\widehat{\wx}_{n-1}),
\end{aligned}\\\label{eq: rtl- E12}\tag{$E_{k\ell}$}
&\begin{aligned}
\phi_{k0}(\widehat{\wx}_n-\whx_n)+\psi_k(\whx_{n+1}-\widehat{\wx}_n)=
  \phi_{\ell 0}(\widehat{\wx}_n-\wx_n)+\psi_\ell(\wx_{n+1}-\widehat{\wx}_n).
\end{aligned}
\end{align}
A visualization of the 2D~corner equations embedded in $\Z^{3}$ is given in Figure~\ref{fig: 2d for Gk, Gl}.
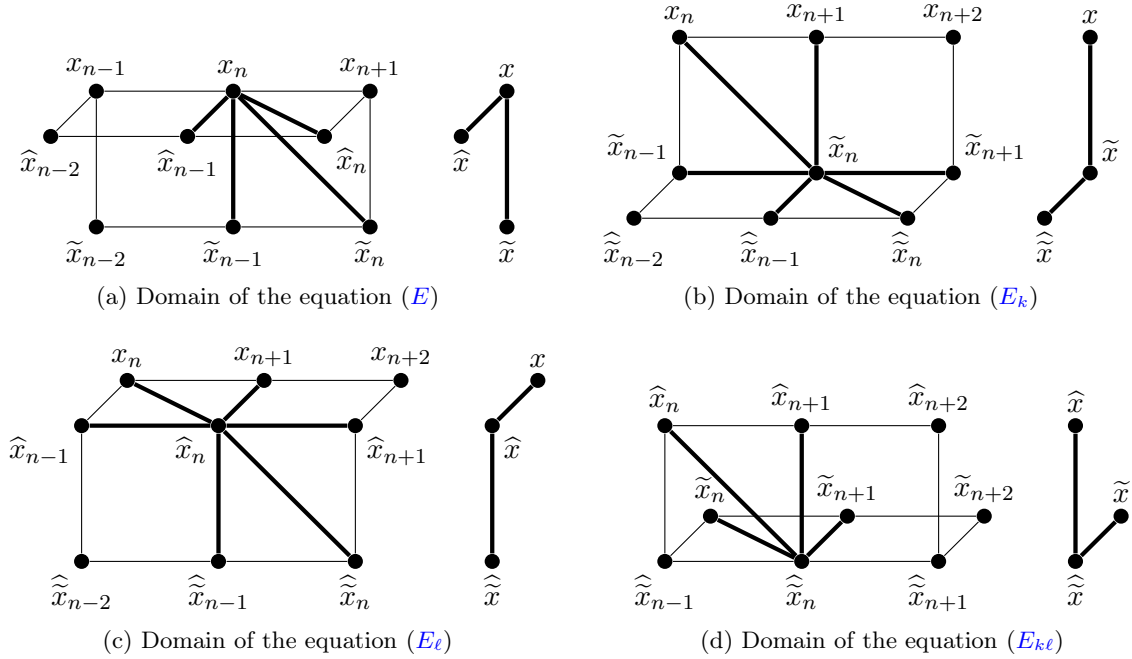
\begin{figure}[htb]
   \centering
      \subfloat[Domain of the equation \eqref{eq: rtl- E}]{\label{fig:7a}
   \begin{tikzpicture}[scale=0.6,inner sep=2]
      \node (x2) at (1,1) [circle,fill,label=-90:$\widetilde{x}_{n-2}$] {};
      \node (x3) at (0,3) [circle,fill,label=-90:$\widehat{x}_{n-2}$] {};
      \node (x12) at (4,1) [circle,fill,label=-90:$\widetilde{x}_{n-1}$] {};
      \node (x13) at (3,3) [circle,fill,label=-90:$\widehat{x}_{n-1}$] {};
      \node (x23) at (1,4) [circle,fill,label=90:$x_{n-1}$] {};
      \node (x112) at (7,1) [circle,fill,label=-90:$\widetilde{x}_{n}$] {};
      \node (x113) at (6,3) [circle,fill,label=-45:$\widehat{x}_{n}$] {};
      \node (x123) at (4,4) [circle,fill,label=90:$x_{n}$] {};
      \node (x1123) at (7,4) [circle,fill,label=90:$x_{n+1}$] {};
      \node (y2) at (10,1) [circle,fill,label=-90:$\widetilde{x}$] {};
      \node (y3) at (9,3) [circle,fill,label=-90:$\widehat{x}$] {};
      \node (y23) at (10,4) [circle,fill,label=90:$x$] {};
      \draw (x2) to (x12);
      \draw (x2) to (x23);
      \draw (x3) to (x13);
      \draw (x3) to (x23);
      \draw (x12) to (x112);
      \draw [ultra thick] (x12) to (x123);
      \draw (x13) to (x113);
      \draw [ultra thick] (x13) to (x123);
      \draw (x23) to (x123);
      \draw [ultra thick] (x112) to (x123);
      \draw (x112) to (x1123);
      \draw [ultra thick] (x113) to (x123);
      \draw (x113) to (x1123);
      \draw (x123) to (x1123);
      \draw [ultra thick] (y2) to (y23) to (y3);
   \end{tikzpicture}
   }\qquad
   \subfloat[Domain of the equation \eqref{eq: rtl- E1}]{\label{fig:7b}
   \begin{tikzpicture}[scale=0.6,inner sep=2]
      \node (x) at (0,0) [circle,fill,label=-90:$\htilde{x}_{n-2}$] {};
      \node (x1) at (3,0) [circle,fill,label=-90:$\htilde{x}_{n-1}$] {};
      \node (x2) at (1,1) [circle,fill,label=135:$\widetilde{x}_{n-1}$] {};
      \node (x11) at (6,0) [circle,fill,label=-90:$\htilde{x}_{n}$] {};
      \node (x12) at (4,1) [circle,fill,label=45:$\widetilde{x}_{n}$] {};
      \node (x23) at (1,4) [circle,fill,label=90:$x_{n}$] {};
      \node (x112) at (7,1) [circle,fill,label=45:$\widetilde{x}_{n+1}$] {};
      \node (x123) at (4,4) [circle,fill,label=90:$x_{n+1}$] {};
      \node (x1123) at (7,4) [circle,fill,label=90:$x_{n+2}$] {};
      \node (y2) at (10,1) [circle,fill,label=45:$\widetilde{x}$] {};
      \node (y) at (9,0) [circle,fill,label=-90:$\htilde{x}$] {};
      \node (y23) at (10,4) [circle,fill,label=90:$x$] {};
      \draw (x) to (x1);
      \draw (x) to (x2);
      \draw (x1) to (x11);
      \draw [ultra thick] (x1) to (x12);
      \draw [ultra thick] (x2) to (x12);
      \draw (x2) to (x23);
      \draw [ultra thick] (x11) to (x12);
      \draw (x11) to (x112);
      \draw [ultra thick] (x12) to (x23);
      \draw [ultra thick] (x12) to (x112);
      \draw [ultra thick] (x12) to (x123);
      \draw (x23) to (x123);
      \draw (x112) to (x1123);
      \draw (x123) to (x1123);
      \draw [ultra thick] (y) to (y2) to (y23);
   \end{tikzpicture}
   }\\
   \subfloat[Domain of the equation \eqref{eq: rtl- E2}]{\label{fig:7c}
   \begin{tikzpicture}[scale=0.6,inner sep=2]
      \node (x) at (0,0) [circle,fill,label=-90:$\htilde{x}_{n-2}$] {};
      \node (x1) at (3,0) [circle,fill,label=-90:$\htilde{x}_{n-1}$] {};
      \node (x3) at (0,3) [circle,fill,label=-135:$\widehat{x}_{n-1}$] {};
      \node (x11) at (6,0) [circle,fill,label=-90:$\htilde{x}_{n}$] {};
      \node (x13) at (3,3) [circle,fill,label=-135:$\widehat{x}_{n}$] {};
      \node (x23) at (1,4) [circle,fill,label=90:$x_{n}$] {};
      \node (x113) at (6,3) [circle,fill,label=-45:$\widehat{x}_{n+1}$] {};
      \node (x123) at (4,4) [circle,fill,label=90:$x_{n+1}$] {};
      \node (x1123) at (7,4) [circle,fill,label=90:$x_{n+2}$] {};
      \node (y) at (9,0) [circle,fill,label=-90:$\htilde{x}$] {};
      \node (y3) at (9,3) [circle,fill,label=-45:$\widehat{x}$] {};
      \node (y23) at (10,4) [circle,fill,label=90:$x$] {};
      \draw (x) to (x1);
      \draw (x) to (x3);
      \draw (x1) to (x11);
      \draw [ultra thick] (x1) to (x13);
      \draw [ultra thick] (x3) to (x13);
      \draw (x3) to (x23);
      \draw [ultra thick] (x11) to (x13);
      \draw (x11) to (x113);
      \draw [ultra thick] (x13) to (x23);
      \draw [ultra thick] (x13) to (x113);
      \draw [ultra thick] (x13) to (x123);
      \draw (x23) to (x123);
      \draw (x113) to (x1123);
      \draw (x123) to (x1123);
      \draw [ultra thick] (y) to (y3) to (y23);
   \end{tikzpicture}
   }\qquad
   \subfloat[Domain of the equation \eqref{eq: rtl- E12}]{\label{fig:7d}
   \begin{tikzpicture}[scale=0.6,inner sep=2]
      \node (x) at (0,0) [circle,fill,label=-90:$\htilde{x}_{n-1}$] {};
      \node (x1) at (3,0) [circle,fill,label=-90:$\htilde{x}_{n}$] {};
      \node (x2) at (1,1) [circle,fill,label=90:$\widetilde{x}_{n}$] {};
      \node (x3) at (0,3) [circle,fill,label=90:$\widehat{x}_{n}$] {};
      \node (x11) at (6,0) [circle,fill,label=-90:$\htilde{x}_{n+1}$] {};
      \node (x12) at (4,1) [circle,fill,label=90:$\widetilde{x}_{n+1}$] {};
      \node (x13) at (3,3) [circle,fill,label=90:$\widehat{x}_{n+1}$] {};
      \node (x112) at (7,1) [circle,fill,label=90:$\widetilde{x}_{n+2}$] {};
      \node (x113) at (6,3) [circle,fill,label=90:$\widehat{x}_{n+2}$] {};
      \node (y2) at (10,1) [circle,fill,label=90:$\widetilde{x}$] {};
      \node (y3) at (9,3) [circle,fill,label=90:$\widehat{x}$] {};
      \node (y) at (9,0) [circle,fill,label=-90:$\htilde{x}$] {};
      \draw (x) to (x1);
      \draw (x) to (x2);
      \draw (x) to (x3);
      \draw [ultra thick] (x1) to (x2);
      \draw [ultra thick] (x1) to (x3);
      \draw (x1) to (x11);
      \draw [ultra thick] (x1) to (x12);
      \draw [ultra thick] (x1) to (x13);
      \draw (x2) to (x12);
      \draw (x3) to (x13);
      \draw (x11) to (x112);
      \draw (x11) to (x113);
      \draw (x12) to (x112);
      \draw (x13) to (x113);
      \draw [ultra thick] (y2) to (y) to (y3);
   \end{tikzpicture}
   }
   \caption{2D~corner equations for commutativity of $G_k$ and $G_\ell$: 2D~corner equations \eqref{eq: rtl- E}, \eqref{eq: rtl- E12} are 3D~corner equations of the corresponding discrete 2-form, while each of the 2D~corner equations \eqref{eq: rtl- E1}, \eqref{eq: rtl- E2} is a sum of two 3D~corner equations sharing one common face.}
   \label{fig: 2d for Gk, Gl}
\end{figure}

\begin{theo}\label{th: rtl-}
Suppose that the fields $x$, $\wx$, and $\whx$ satisfy 2D~corner equations \eqref{eq: rtl- E}. Define the fields $\widehat{\wx}$ by any of the following four formulae, which are equivalent by virtue of \eqref{eq: rtl- E}:
\begin{align}
\label{eq: rtl- S1}\tag{$S1$}
& \psi_k(x_{n+1}-\wx_n)+\phi_{\ell k}(\whx_n-\wx_n)=\psi_0(\wx_{n+1}-\wx_n)+\phi_{\ell 0}(\widehat{\wx}_n-\wx_n), \\
\label{eq: rtl- S2}\tag{$S2$}
& \psi_\ell(x_{n+1}-\whx_n)+\phi_{k\ell}(\wx_n-\whx_n)=\psi_0(\whx_{n+1}-\whx_n)+\phi_{k0}(\widehat{\wx}_n-\whx_n), \\
\label{eq: rtl- S4}\tag{$S3$}
& \psi_k(\whx_n-\widehat{\wx}_{n-1})+\phi_{\ell k}(\whx_n-\wx_n)=\psi_0(\whx_n-\whx_{n-1})+\phi_{\ell 0}(\whx_n-x_n), \\
\label{eq: rtl- S3}\tag{$S4$}
& \psi_\ell(\wx_n-\widehat{\wx}_{n-1})+\phi_{k\ell}(\wx_n-\whx_n)=\psi_0(\wx_n-\wx_{n-1})+\phi_{k0}(\wx_n-x_n).
\end{align}
called superposition formulae. Then 2D~corner equations \eqref{eq: rtl- E1}--\eqref{eq: rtl- E12} are satisfied, as well.
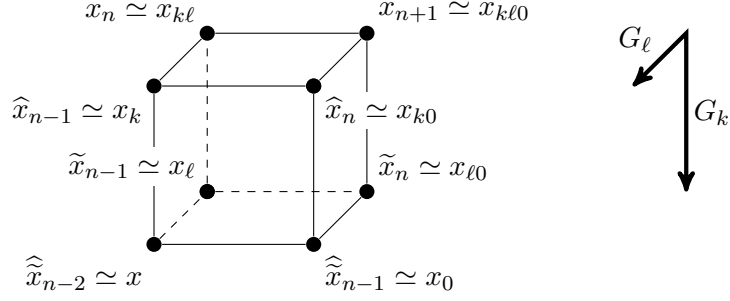
\begin{figure}[htb]
   \centering
   \begin{tikzpicture}[auto,scale=0.7,inner sep=2,>=stealth']
      \node (x) at (0,0) [circle,fill,label=-135:{$\htilde{x}_{n-2}\simeq x$}] {};
      \node (x1) at (3,0) [circle,fill,label=-45:{$\htilde{x}_{n-1}\simeq x_{0}$}] {};
      \node (x2) at (1,1) [circle,fill,label=135:{$\widetilde{x}_{n-1}\simeq x_{\ell}$}] {};
      \node (x3) at (0,3) [circle,fill,label=-135:{$\widehat{x}_{n-1}\simeq x_{k}$}] {};
      \node (x12) at (4,1) [circle,fill,label=45:{$\widetilde{x}_{n}\simeq x_{\ell0}$}] {};
      \node (x23) at (1,4) [circle,fill,label=135:{$x_{n}\simeq x_{k\ell}$}] {};
      \node (x13) at (3,3) [circle,fill,label=-45:{$\widehat{x}_{n}\simeq x_{k0}$}] {};
      \node (x123) at (4,4) [circle,fill,label=45:{$x_{n+1}\simeq x_{k\ell0}$}] {};
      \draw (x) to (x1) to (x13) to (x3) to (x);
      \draw (x1) to (x12) to (x123) to (x23) to (x3);
      \draw (x13) to (x123);
      \draw [dashed] (x) to (x2) to (x12);
      \draw [dashed] (x2) to (x23);
      \node (x13) at (3,3) [circle,fill,label={-45,fill=white:{$\widehat{x}_{n}\simeq x_{k0}$}}] {};
      \node (x2) at (1,1) [circle,fill,label={135,fill=white:{$\widetilde{x}_{n-1}\simeq x_{\ell}$}}] {};
      \draw [ultra thick,<->] (9,3) to node {$G_{\ell}$} (10,4) to node {$G_{k}$} (10,1);
   \end{tikzpicture}
   \caption{Identification of fields: the 0\textsuperscript{th} coordinate direction enumerates the lattice cites, the coordinate directions $k$, $\ell$ correspond to the maps $G_k$, $G_\ell$.}
   \label{fig: 3d for Gk, Gl}
\end{figure}
\begin{proof}
We identify the fields as on Figure~\ref{fig: 3d for Gk, Gl}. Then equations \eqref{eq: rtl- E}, \eqref{eq: rtl- E12} and \eqref{eq: rtl- S1}--\eqref{eq: rtl- S4} build the system of consistent 3D~corner equations \eqref{eq: 3point Ei}, \eqref{eq: 3point Eij}. More precisely, the correspondence is as follows:
\begin{itemize}
\item The downshifted version of \eqref{eq: rtl- E} is $({\mathcal E}_{k\ell})$.
\item The downshifted version of \eqref{eq: rtl- E12} is $({\mathcal E}_0)$.
\item The downshifted versions of \eqref{eq: rtl- S1} and \eqref{eq: rtl- S2} are $({\mathcal E}_\ell)$ and $({\mathcal E}_k)$, respectively.
\item Equations \eqref{eq: rtl- S3} and \eqref{eq: rtl- S4} are $({\mathcal E}_{0\ell})$ and $({\mathcal E}_{0k})$, respectively.
\end{itemize}
Since the system of 3D~corner equations is consistent and, therefore, has rank 2, the following argumentation works: if \eqref{eq: rtl- E} and one of equations \eqref{eq: rtl- S1}--\eqref{eq: rtl- S4} are satisfied, then equation \eqref{eq: rtl- E12} and the remaining three equations of \eqref{eq: rtl- S1}--\eqref{eq: rtl- S4} are fulfilled, as well. Furthermore, equation \eqref{eq: rtl- E1} is a difference of \eqref{eq: rtl- S1} and \eqref{eq: rtl- S4}, and equation \eqref{eq: rtl- E2} is a difference of \eqref{eq: rtl- S2} and \eqref{eq: rtl- S3}. This completes the proof.
\end{proof}
\end{theo}

\subsection{Proof of commutativity of the maps \texorpdfstring{$F_i$, $G_\ell$}{F\_i, G\_l}}
\label{sect: rtl+-}

The 2D~corner equations for the pluri-Lagrangian system corresponding to the maps $F_i$ and $G_\ell$ (whose actions are encoded by $\widetilde{\hspace{5pt}}$ and $\widehat{\hspace{5pt}}$, respectively) are given by:
\begin{align}
\label{eq: rtl+- E}\tag{$E$}
& \begin{aligned}
&\psi_i(\wx_n-x_n)+\phi_{i0}(x_n-\wx_{n-1})-\psi_0(x_{n+1}-x_n)+\psi_0(x_n-x_{n-1})\\
&\hspace{7cm} =\phi_{\ell 0}(\whx_n-x_n)+\psi_\ell(x_n-\whx_{n-1}),
\end{aligned}\\
\label{eq: rtl+- E2}\tag{$E_i$}
&\begin{aligned}
&\psi_i(\wx_n-x_n)+\phi_{i0}(x_{n+1}-\wx_n)=
\phi_{\ell 0}(\widehat{\wx}_n-\wx_n)+\psi_\ell(\wx_n-\widehat{\wx}_{n-1}),
\end{aligned}\\\label{eq: rtl+- E1}\tag{$E_\ell$}
&\begin{aligned}
\psi_i(\widehat{\wx}_n-\whx_n)+\phi_{i0}(\whx_n-\widehat{\wx}_{n-1})=
\phi_{\ell 0}(\whx_n-x_n)+\psi_\ell(x_{n+1}-\whx_n),
\end{aligned}\\\label{eq: rtl+- E12}\tag{$E_{i\ell}$}
&\begin{aligned}
&\psi_i(\widehat{\wx}_n-\whx_n)+\phi_{i0}(\whx_{n+1}-\widehat{\wx}_n)=\phi_{\ell 0}(\widehat{\wx}_n-\wx_n)+\psi_\ell(\wx_{n+1}-\widehat{\wx}_n)\\
&\hspace{7cm}-\psi_0(\widehat{\wx}_{n+1}-\widehat{\wx}_n)+\psi_0(\widehat{\wx}_n-\widehat{\wx}_{n-1}).
\end{aligned}
\end{align}

A visualization of the 2D~corner equations embedded in $\Z^{3}$ is given in Figure~\ref{fig:8}.
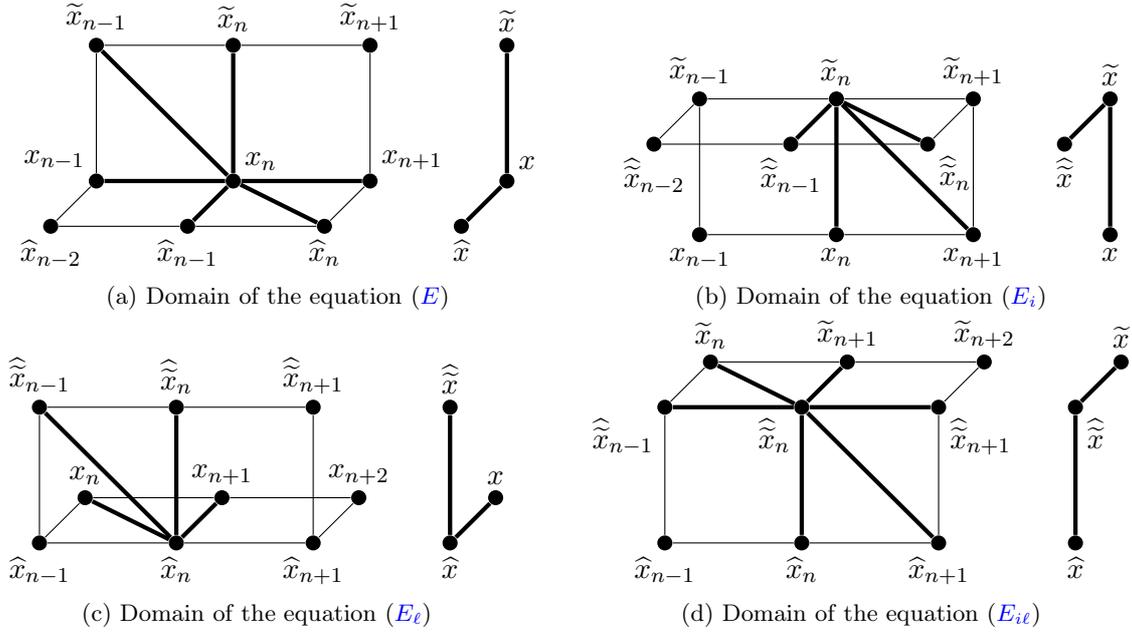
\begin{figure}[htb]
   \centering
   \subfloat[Domain of the equation \eqref{eq: rtl+- E}]{\label{fig:8a}
   \begin{tikzpicture}[scale=0.6,inner sep=2]
      \node (x) at (0,0) [circle,fill,label=-90:$\whx_{n-2}$] {};
      \node (x1) at (3,0) [circle,fill,label=-90:$\whx_{n-1}$] {};
      \node (x2) at (1,1) [circle,fill,label=135:$x_{n-1}$] {};
      \node (x11) at (6,0) [circle,fill,label=-90:$\whx_{n}$] {};
      \node (x12) at (4,1) [circle,fill,label=45:$x_{n}$] {};
      \node (x23) at (1,4) [circle,fill,label=90:$\wx_{n-1}$] {};
      \node (x112) at (7,1) [circle,fill,label=45:$x_{n+1}$] {};
      \node (x123) at (4,4) [circle,fill,label=90:$\wx_{n}$] {};
      \node (x1123) at (7,4) [circle,fill,label=90:$\wx_{n+1}$] {};
      \node (y) at (9,0) [circle,fill,label=-90:$\whx$] {};
      \node (y2) at (10,1) [circle,fill,label=45:$x$] {};
      \node (y23) at (10,4) [circle,fill,label=90:$\wx$] {};
      \draw (x) to (x1);
      \draw (x) to (x2);
      \draw (x1) to (x11);
      \draw [ultra thick] (x1) to (x12);
      \draw [ultra thick] (x2) to (x12);
      \draw (x2) to (x23);
      \draw [ultra thick] (x11) to (x12);
      \draw (x11) to (x112);
      \draw [ultra thick] (x12) to (x23);
      \draw [ultra thick] (x12) to (x112);
      \draw [ultra thick] (x12) to (x123);
      \draw (x23) to (x123);
      \draw (x112) to (x1123);
      \draw (x123) to (x1123);
      \draw [ultra thick] (y) to (y2) to (y23);
   \end{tikzpicture}
   }\qquad
   \subfloat[Domain of the equation \eqref{eq: rtl+- E2}]{\label{fig:8b}
   \begin{tikzpicture}[scale=0.6,inner sep=2]
      \node (x2) at (1,1) [circle,fill,label=-90:$x_{n-1}$] {};
      \node (x3) at (0,3) [circle,fill,label=-90:$\widehat{\wx}_{n-2}$] {};
      \node (x12) at (4,1) [circle,fill,label=-90:$x_{n}$] {};
      \node (x13) at (3,3) [circle,fill,label=-90:$\widehat{\wx}_{n-1}$] {};
      \node (x23) at (1,4) [circle,fill,label=90:$\wx_{n-1}$] {};
      \node (x112) at (7,1) [circle,fill,label=-90:$x_{n+1}$] {};
      \node (x113) at (6,3) [circle,fill,label=-45:$\widehat{\wx}_{n}$] {};
      \node (x123) at (4,4) [circle,fill,label=90:$\wx_{n}$] {};
      \node (x1123) at (7,4) [circle,fill,label=90:$\wx_{n+1}$] {};
      \node (y3) at (9,3) [circle,fill,label=-90:$\widehat{\wx}$] {};
      \node (y2) at (10,1) [circle,fill,label=-90:$x$] {};
      \node (y23) at (10,4) [circle,fill,label=90:$\wx$] {};
      \draw (x2) to (x12);
      \draw (x2) to (x23);
      \draw (x3) to (x13);
      \draw (x3) to (x23);
      \draw (x12) to (x112);
      \draw [ultra thick] (x12) to (x123);
      \draw (x13) to (x113);
      \draw [ultra thick] (x13) to (x123);
      \draw (x23) to (x123);
      \draw [ultra thick] (x112) to (x123);
      \draw (x112) to (x1123);
      \draw [ultra thick] (x113) to (x123);
      \draw (x113) to (x1123);
      \draw (x123) to (x1123);
      \draw [ultra thick] (y2) to (y23) to (y3);
   \end{tikzpicture}
   }\\
   \subfloat[Domain of the equation \eqref{eq: rtl+- E1}]{\label{fig:8c}
   \begin{tikzpicture}[scale=0.6,inner sep=2]
      \node (x) at (0,0) [circle,fill,label=-90:$\whx_{n-1}$] {};
      \node (x1) at (3,0) [circle,fill,label=-90:$\whx_{n}$] {};
      \node (x2) at (1,1) [circle,fill,label=90:$x_{n}$] {};
      \node (x3) at (0,3) [circle,fill,label=90:$\widehat{\wx}_{n-1}$] {};
      \node (x11) at (6,0) [circle,fill,label=-90:$\whx_{n+1}$] {};
      \node (x12) at (4,1) [circle,fill,label=90:$x_{n+1}$] {};
      \node (x13) at (3,3) [circle,fill,label=90:$\widehat{\wx}_{n}$] {};
      \node (x112) at (7,1) [circle,fill,label=90:$x_{n+2}$] {};
      \node (x113) at (6,3) [circle,fill,label=90:$\widehat{\wx}_{n+1}$] {};
      \node (y) at (9,0) [circle,fill,label=-90:$\whx$] {};
      \node (y2) at (10,1) [circle,fill,label=90:$x$] {};
      \node (y3) at (9,3) [circle,fill,label=90:$\widehat{\wx}$] {};
      \draw (x) to (x1);
      \draw (x) to (x2);
      \draw (x) to (x3);
      \draw [ultra thick] (x1) to (x2);
      \draw [ultra thick] (x1) to (x3);
      \draw (x1) to (x11);
      \draw [ultra thick] (x1) to (x12);
      \draw [ultra thick] (x1) to (x13);
      \draw (x2) to (x12);
      \draw (x3) to (x13);
      \draw (x11) to (x112);
      \draw (x11) to (x113);
      \draw (x12) to (x112);
      \draw (x13) to (x113);
      \draw [ultra thick] (y2) to (y) to (y3);
   \end{tikzpicture}
   }\qquad
   \subfloat[Domain of the equation \eqref{eq: rtl+- E12}]{\label{fig:8d}
   \begin{tikzpicture}[scale=0.6,inner sep=2]
      \node (x) at (0,0) [circle,fill,label=-90:$\whx_{n-1}$] {};
      \node (x1) at (3,0) [circle,fill,label=-90:$\whx_{n}$] {};
      \node (x3) at (0,3) [circle,fill,label=-135:$\widehat{\wx}_{n-1}$] {};
      \node (x11) at (6,0) [circle,fill,label=-90:$\whx_{n+1}$] {};
      \node (x13) at (3,3) [circle,fill,label=-135:$\widehat{\wx}_{n}$] {};
      \node (x23) at (1,4) [circle,fill,label=90:$\wx_{n}$] {};
      \node (x113) at (6,3) [circle,fill,label=-45:$\widehat{\wx}_{n+1}$] {};
      \node (x123) at (4,4) [circle,fill,label=90:$\wx_{n+1}$] {};
      \node (x1123) at (7,4) [circle,fill,label=90:$\wx_{n+2}$] {};
      \node (y) at (9,0) [circle,fill,label=-90:$\whx$] {};
      \node (y23) at (10,4) [circle,fill,label=90:$\wx$] {};
      \node (y3) at (9,3) [circle,fill,label=-45:$\widehat{\wx}$] {};
      \draw (x) to (x1);
      \draw (x) to (x3);
      \draw (x1) to (x11);
      \draw [ultra thick] (x1) to (x13);
      \draw [ultra thick] (x3) to (x13);
      \draw (x3) to (x23);
      \draw [ultra thick] (x11) to (x13);
      \draw (x11) to (x113);
      \draw [ultra thick] (x13) to (x23);
      \draw [ultra thick] (x13) to (x113);
      \draw [ultra thick] (x13) to (x123);
      \draw (x23) to (x123);
      \draw (x113) to (x1123);
      \draw (x123) to (x1123);
      \draw [ultra thick] (y) to (y3) to (y23);
   \end{tikzpicture}
   }
   \caption{2D~corner equations for the maps $F_i$ and $G_\ell$: 2D~corner equations \eqref{eq: rtl+- E2}, \eqref{eq: rtl+- E1} are 3D~corner equations of the corresponding discrete 2-form, while each one of the 2D~corner equations \eqref{eq: rtl+- E}, \eqref{eq: rtl+- E12} is a sum of two 3D~corner equations sharing one common face.}
   \label{fig:8}
\end{figure}

\begin{theo}\label{th: rtl+-}
Suppose that the fields $x$, $\wx$, and $\whx$ satisfy 2D~corner equations \eqref{eq: rtl+- E}. Define the fields $\widehat{\wx}$ by any of the following two formulae, which are equivalent by virtue of \eqref{eq: rtl+- E}:
\begin{align}
\label{eq: rtl+- S1}\tag{$S1$}
& \psi_i(\wx_n-x_n)+\phi_{i\ell}(x_n-\widehat{\wx}_{n-1})=\phi_{\ell 0}(\whx_n-x_n)+\psi_0(x_{n+1}-x_n), \\
\label{eq: rtl+- S2}\tag{$S2$}
& \psi_0(x_{n+1}-x_n)+\phi_{i0}(x_{n+1}-\wx_n)=
\psi_\ell(x_{n+1}-\whx_n)+\phi_{i\ell}(x_{n+1}-\widehat{\wx}_n),
\end{align}
called superposition formulae. Then the 2D~corner equations \eqref{eq: rtl+- E2}, \eqref{eq: rtl+- E1}, \eqref{eq: rtl+- E12} and the following two equations are satisfied, as well:
\begin{align}
\label{eq: rtl+- S3}\tag{$S3$}
& \psi_0(\widehat{\wx}_n-\widehat{\wx}_{n-1})+\phi_{i0}(\whx_n-\widehat{\wx}_{n-1})
  =\psi_\ell(\wx_n-\widehat{\wx}_{n-1})+\phi_{i\ell}(x_n-\widehat{\wx}_{n-1}), \\
\label{eq: rtl+- S4}\tag{$S4$}
& \psi_i(\widehat{\wx}_n-\whx_n)+\phi_{i\ell}(x_{n+1}-\widehat{\wx}_n)=
  \phi_{\ell 0}(\widehat{\wx}_n-\wx_n)+\psi_0(\widehat{\wx}_n-\widehat{\wx}_{n-1}).
\end{align}

\begin{figure}[htb]
   \centering
   \begin{tikzpicture}[auto,scale=0.7,inner sep=2,>=stealth']
      \node (x) at (0,0) [circle,fill,label=-135:{$\widehat{x}_{n-1}\simeq x$}] {};
      \node (x1) at (3,0) [circle,fill,label=-45:{$\widehat{x}_{n}\simeq x_{0}$}] {};
      \node (x2) at (1,1) [circle,fill,label=135:{$x_{n}\simeq x_{\ell}$}] {};
      \node (x3) at (0,3) [circle,fill,label=-135:{$\htilde{x}_{n-1}\simeq x_{i}$}] {};
      \node (x12) at (4,1) [circle,fill,label=45:{$x_{n+1}\simeq x_{\ell0}$}] {};
      \node (x23) at (1,4) [circle,fill,label=135:{$\widetilde{x}_{n}\simeq x_{i\ell}$}] {};
      \node (x13) at (3,3) [circle,fill,label=-45:{$\htilde{x}_{n}\simeq x_{i0}$}] {};
      \node (x123) at (4,4) [circle,fill,label=45:{$\widetilde{x}_{n+1}\simeq x_{i\ell0}$}] {};
      \draw (x) to (x1) to (x13) to (x3) to (x);
      \draw (x1) to (x12) to (x123) to (x23) to (x3);
      \draw (x13) to (x123);
      \draw [dashed] (x) to (x2) to (x12);
      \draw [dashed] (x2) to (x23);
      \node (x13) at (3,3) [circle,fill,label={-45,fill=white:{$\htilde{x}_{n}\simeq x_{i0}$}}] {};
      \node (x2) at (1,1) [circle,fill,label={135,fill=white:{$x_{n}\simeq x_{\ell}$}}] {};
      \draw [ultra thick,<->] (9,0) to node [swap] {$G_{\ell}$} (10,1) to node [swap] {$F_{i}$} (10,4);
   \end{tikzpicture}
   \caption{Identification of fields: the 0\textsuperscript{th} coordinate direction enumerates the lattice cites, the coordinate directions $i$, $\ell$ correspond to the maps $F_i$, $G_\ell$, respectively.}
   \label{fig:10}
\end{figure}
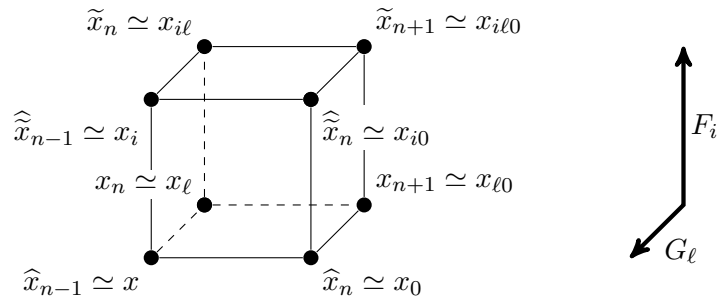

\begin{proof}
We see that the present theorem works in the same manner as Theorems \ref{th: rtl+} and \ref{th: rtl-}, with the only difference that now we only have two local superposition formulae \eqref{eq: rtl+- S1} and \eqref{eq: rtl+- S2}.
We identify the fields in the way which is in Figure~\ref{fig:10}. Then equations \eqref{eq: rtl+- E1}, \eqref{eq: rtl+- E2} and \eqref{eq: rtl+- S1}--\eqref{eq: rtl+- S4} build the system of consistent 3D~corner equations \eqref{eq: 3point Ei}, \eqref{eq: 3point Eij}. More precisely, the correspondence is as follows:\par
\newlength{\firstcol}
\setlength{\firstcol}{0.5\textwidth}
\addtolength{\firstcol}{-34.88786pt}
\begin{tabular}{p{\firstcol}p{0.5\textwidth}}
\begin{itemize}
\item Equation \eqref{eq: rtl+- E1} is $({\mathcal E}_0)$.
\item Equation \eqref{eq: rtl+- E2} is $({\mathcal E}_{i\ell})$.
\item Equation \eqref{eq: rtl+- S1} is $({\mathcal E}_\ell)$.
\end{itemize}&
\begin{itemize}
\item Equation \eqref{eq: rtl+- S2} is $({\mathcal E}_{0\ell})$.
\item Equation \eqref{eq: rtl+- S3} is $({\mathcal E}_i)$.
\item Equation \eqref{eq: rtl+- S4} is $({\mathcal E}_{0i})$.
\end{itemize}
\end{tabular}\par
First, we observe that equation \eqref{eq: rtl+- E} is a sum of equation \eqref{eq: rtl+- S1} and (a downshifted version of) equation \eqref{eq: rtl+- S2}. Therefore, if \eqref{eq: rtl+- E} and one of equations \eqref{eq: rtl+- S1}, \eqref{eq: rtl+- S2} hold, the remaining equations of \eqref{eq: rtl+- S1}, \eqref{eq: rtl+- S2} holds, too.

Due to the fact that the system of 3D~corner equations has rank 2, we can claim the following: if \eqref{eq: rtl+- E} and one of equations \eqref{eq: rtl+- S1}, \eqref{eq: rtl+- S2} is satisfied, then equations \eqref{eq: rtl+- E1}, \eqref{eq: rtl+- E2}, \eqref{eq: rtl+- S3} and \eqref{eq: rtl+- S4} are fulfilled, as well.

Finally, we observe that equation \eqref{eq: rtl+- E12} is a sum of (an upshifted version of) equation \eqref{eq: rtl+- S3} and equation \eqref{eq: rtl+- S4}. This completes the proof.
\end{proof}
\end{theo}

\section{B\"acklund transformations for symmetric systems of the relativistic Toda type}
\label{sect: sym}

We start with the pluri-Lagrangian systems related to the quad-equation $Q_{1}^{0}$, $Q_1^1$, and $Q_3^0$ from the ABS~list (see~\cite{octahedron} for further information). Each of these systems is constructed with the help of just one fundamental function $\phi(x)$, which is given for these three cases by
\begin{equation}\label{eq: sym phi}
\phi(x;\alpha)=\frac{\alpha}{x},\qquad \phi(x;\alpha)=\frac{1}{2}\log\frac{x+\alpha}{x-\alpha},\qquad\text{and}\qquad
\phi(x;\alpha)=\frac{1}{2}\log\frac{\sinh(x+\alpha)}{\sinh(x-\alpha)},
\end{equation}
respectively. The corner equations are given by \eqref{eq: 3point Ei}, \eqref{eq: 3point Eij} with the leg functions
\[
\psi_i(x)=\phi(x;\alpha_i),\qquad
\phi_{ij}(x)=\phi(x;\alpha_i-\alpha_j),
\]
as well as the following choice of parameters:
\[
\alpha_0=\alpha,\qquad
\alpha_i=\lambda,\qquad
\alpha_j=\mu,\qquad
\alpha_k=\lambda+\alpha,\qquad
\alpha_\ell=\mu+\alpha.
\]
Thus, the two mutually commuting families of B\"acklund transformations are given by
\begin{equation}\label{eq: dRTL1 map}
    F_{i}: \begin{cases}[2]
\ p_n=\phi(\wx_n-x_n;\lambda)+\phi(x_n-\wx_{n-1};\lambda-\alpha)
-\phi(x_{n+1}-x_n;\alpha)+\phi(x_n-x_{n-1};\alpha),\\
\ \widetilde{p}_n=\phi(\wx_n-x_n;\lambda)+\phi(x_{n+1}-\wx_n;\lambda-\alpha),
\end{cases}
\end{equation}
and
\begin{equation}\label{eq: dRTL2 map}
G_{k}: \begin{cases}[2]
\ p_n=\phi(\wx_n-x_n;\lambda)+\phi(x_n-\wx_{n-1};\lambda+\alpha),\\
\ \widetilde{p}_n=\phi(\wx_n-x_n;\lambda)+\phi(x_{n+1}-\wx_n;\lambda+\alpha)-
\phi(\wx_{n+1}-\wx_n;\alpha)+\phi(\wx_n-\wx_{n-1};\alpha).
\end{cases}
\end{equation}
The functions $\phi_{ij}(x)$, $\phi_{k\ell}(x)$, and $\phi_{i\ell}(x)$ used in Theorems~\ref{th: rtl+}, \ref{th: rtl-}, and \ref{th: rtl+-} to prove commutativity of any two of the maps $F_i$, $F_j$, $G_k$, and $G_\ell$, are given by
\[
\phi_{ij}(x)=\phi(x;\lambda-\mu),\qquad \phi_{k\ell}(x)=\phi(x;\lambda-\mu),\qquad\text{and}\qquad
\phi_{i\ell}(x)=\phi(x;\lambda-\mu-\alpha).
\]

\section{B\"acklund transformations for the modified exponential system of the relativistic Toda type}\label{sect: mod Toda}

The pluri-Lagrangian system playing the ``master'' role of the  system for all the asymmetric systems of the relativistic Toda type is described in the following proposition.

\begin{prop}\label{th: 3Dcorners asym}
The pluri-Lagrangian system consisting of the 3D~corner equations \eqref{eq: 3point Ei}, \eqref{eq: 3point Eij} with the leg functions
\begin{equation}\label{eq:master legs}
\psi_j(x)=\log\frac{\alpha_j+e^x}{\beta_j}, \qquad
\phi_{ij}(x)=\log\frac{\beta_j-\beta_i e^x}{\alpha_j-\alpha_i e^x}
\end{equation}
is consistent, and the corresponding discrete Lagrangian 2-form $\Ell$ is closed on its solutions.
\end{prop}
\begin{proof}
The system in question is a slight generalization of a similar pluri-Lagrangian system for which the analogous statements were proven in \cite{variational}, and which consists of the 3D~corner equations \eqref{eq: 3point Ei}, \eqref{eq: 3point Eij} with the leg functions
\begin{equation}\label{eq: asym legs}
\bar\psi_j(x)=\log(\gamma_j-e^x), \qquad
\bar\phi_{ij}(x)=\log\frac{\gamma_i-\gamma_j e^x}{\gamma_j-\gamma_i e^x}.
\end{equation}
Setting
\begin{equation}
x_i=\bar{x}_i+\log(-\gamma_i\beta_i),\qquad
x_{ij}=\bar{x}_{ij}+\log(\gamma_i\gamma_j\beta_i\beta_j),\qquad
\text{and}\qquad
\alpha_i=\gamma_i^2\beta_i,
\end{equation}
we find:
\begin{equation}\label{eq: x psi shifts}
\psi_j(x_{ij}-x_i)=\bar{\psi}_j(\bar{x}_{ij}-\bar{x}_i)+\log\gamma_j,
\qquad
\phi_{ij}(x_j-x_i)=\bar{\phi}_{ij}(\bar{x}_j-\bar{x}_i)-\log(\gamma_i\gamma_j).
\end{equation}
This yields, up to additive constants,
\[
L_j(x_{ij}-x_i)=\bar{L}_j(\bar{x}_{ij}-\bar{x}_i)+\log\gamma_j(\bar{x}_{ij}-\bar{x}_i),
\]
\[
\Lambda_{ij}(x_j-x_i)=\bar{\Lambda}_{ij}(\bar{x}_j-\bar{x}_i)-\log(\gamma_i\gamma_j)(\bar{x}_j-\bar{x}_i).
\]
Now, the closure relation for the system with the leg functions \eqref{eq:master legs} is immediately seen to be equivalent to the closure relation for the system with the leg functions \eqref{eq: asym legs}.
\end{proof}

In identifying the leg functions, we will often silently perform shifts similar to \eqref{eq: x psi shifts},
\begin{equation}\label{eq: psi shifts}
\psi_j \rightsquigarrow \psi_j+\epsilon_j, \qquad \phi_{ij}\rightsquigarrow \phi_{ij}-\epsilon_i-\epsilon_j,
\end{equation}
which affect neither the 3D-corner equations nor the closedness of the pluri-Lagrangian 2-form $\Ell$.

We start with the following choice of parameters:
\begin{alignat*}{5}
&\alpha_0=\frac{1}{\alpha},&\qquad
&\alpha_i=-1,&\qquad
&\alpha_j=-1,&\qquad
&\alpha_{k}=\frac{1}{\lambda+\alpha},&\qquad
&\alpha_{\ell}=\frac{1}{\mu+\alpha},\\
&\beta_0=-1,&\qquad
&\beta_i=\lambda,&\qquad
&\beta_j=\mu,&\qquad
&\beta_{k}=-1,&\qquad
&\beta_{\ell}=-1.
\end{alignat*}
With this choice of parameters, the relevant leg functions for the map $F_i$ are
\[
\psi_0(x)=\log(1+\alpha e^x),\qquad \psi_i(x)=\log\frac{e^{x}-1}{\lambda},\qquad
\phi_{i0}(x)=\log\frac{1+\lambda e^x}{1+\alpha e^x},
\]
and those relevant for the map $G_k$ are
\[
\psi_k(x)=\log(1+(\lambda+\alpha)e^x),\qquad
\phi_{k0}(x)=\log\frac{\lambda^{-1}(e^x-1)}{1-\lambda^{-1}\alpha(e^x-1)}.
\]
Moreover, we always assume that the functions $\psi_j$, $\phi_{j0}$ are obtained from $\psi_i$, $\phi_{i0}$ by replacing $\lambda$ by $\mu$, and, the functions $\psi_\ell$, $\phi_{\ell 0}$ are obtained from $\psi_k$, $\phi_{k0}$ in the same way. Thus, we are dealing with the following two families of B\"acklund transformations:
\begin{align}
\label{eq: mod Toda F}
F_{i}:&
\begin{cases}[2]
e^{p_{n}}=\dfrac{e^{\widetilde{x}_{n}-x_{n}}-1}{\lambda}\cdot
\dfrac{1+\lambda e^{x_{n}-\widetilde{x}_{n-1}}}{1+\alpha e^{x_{n}-\widetilde{x}_{n-1}}}\cdot
\dfrac{1+\alpha e^{x_{n}-x_{n-1}}}{1+\alpha e^{x_{n+1}-x_{n}}}, \\
e^{\widetilde{p}_{n}}=\dfrac{e^{\widetilde{x}_{n}-x_{n}}-1}{\lambda}\cdot
\dfrac{1+\lambda e^{x_{n}+1-\widetilde{x}_{n}}}{1+\alpha e^{x_{n+1}-\widetilde{x}_{n}}};
\end{cases}\\
\intertext{and}
\label{eq: mod Toda G}
G_{k}:&
\begin{cases}[2]
e^{p_{n}}=
\dfrac{\lambda^{-1}(e^{\widetilde{x}_{n}-x_{n}}-1)}{1-\lambda^{-1}\alpha(e^{\widetilde{x}_{n}-x_{n}}-1)}\cdot
(1+(\lambda+\alpha)e^{x_{n}-\widetilde{x}_{n-1}}),\\
e^{\widetilde{p}_{n}}=
\dfrac{\lambda^{-1}(e^{\widetilde{x}_{n}-x_{n}}-1)}{1-\lambda^{-1}\alpha(e^{\widetilde{x}_{n}-x_{n}}-1)}\cdot
(1+(\lambda+\alpha)e^{x_{n+1}-\widetilde{x}_{n}})\cdot
\dfrac{1+\alpha e^{\widetilde{x}_{n}-\widetilde{x}_{n-1}}}{1+\alpha e^{\widetilde{x}_{n+1}-\widetilde{x}_{n}}}.
\end{cases}
\end{align}
The functions $\phi_{ij}(x)$, $\phi_{k\ell}(x)$, and $\phi_{i\ell}(x)$ used in Theorems~\ref{th: rtl+}, \ref{th: rtl-}, and \ref{th: rtl+-} to prove commutativity of any two of the maps $F_i$, $F_j$, $G_k$, and $G_\ell$, are given by
\[
\phi_{ij}(x)=\log\frac{\lambda e^x-\mu}{e^x-1},\qquad
\phi_{k\ell}(x)=\log\frac{e^x-1}{(\mu+\alpha)e^x-(\lambda+\alpha)},\qquad
\phi_{i\ell}(x)=\log\frac{1+\lambda e^x}{1+(\mu+\alpha)e^x}.
\]

\section{B\"acklund transformations for the ``master'' exponential system of the relativistic Toda type}
\label{sect: master}

The following system is algebraically similar to the modified exponential one, but provides us with more freedom in the choice of parameters. We set in the system of Proposition \ref{th: 3Dcorners asym}:
\begin{alignat*}{5}
&\alpha_0=\frac{1}{\epsilon\alpha},&\qquad
&\alpha_i=\frac{\lambda-\epsilon}{\epsilon},&\qquad
&\alpha_j=\frac{\mu-\epsilon}{\epsilon},&\qquad
&\alpha_{k}=\frac{1}{\epsilon(\alpha+\lambda)},&\qquad
&\alpha_{\ell}=\frac{1}{\epsilon(\alpha+\mu)},\\
&\beta_0=\frac{1}{\alpha-\epsilon},&\qquad
&\beta_i=\lambda,&\qquad
&\beta_j=\mu,&\qquad
&\beta_{k}=\frac{1}{\alpha+\lambda-\epsilon},&\qquad
&\beta_{\ell}=\frac{1}{\alpha+\mu-\epsilon}.
\end{alignat*}
Up to the shifts of the type \eqref{eq: psi shifts}, we find:
\begin{align*}
&\psi_0(x)=\log(1+\epsilon\alpha e^x),\qquad
\psi_i(x)=\log\left(1+\frac{\epsilon}{\lambda}(e^x-1)\right),\qquad
\phi_{i0}(x)=\log\frac{1-\lambda(\alpha-\epsilon)e^x}{1-\alpha(\lambda-\epsilon)e^x},\\
\intertext{and}
&\psi_k(x)=\log(1+\epsilon(\lambda+\alpha)e^x),\qquad
\phi_{k0}(x)=\log\frac{1-\lambda^{-1}(\alpha-\epsilon)(e^x-1)}{1-\lambda^{-1}\alpha(e^x-1)},
\end{align*}
so that we are dealing with the following two families of B\"acklund transformations:
\begin{align}
F_{i}:&
\begin{cases}[2]
e^{\epsilon p_n}=
\left(1+\epsilon\lambda^{-1}\left(e^{\wx_n-x_n}-1\right)\right)\cdot
\dfrac{1-\lambda(\alpha-\epsilon)e^{x_n-\wx_{n-1}}}{1-\alpha(\lambda-\epsilon)e^{x_n-\wx_{n-1}}}\cdot
\dfrac{1+\epsilon\alpha e^{x_n-x_{n-1}}}{1+\epsilon\alpha e^{x_{n+1}-x_n}},\\
e^{\epsilon\widetilde{p}_n}=\left(1+\epsilon\lambda^{-1}\left(e^{\wx_n-x_n}-1\right)\right)\cdot
\dfrac{1-\lambda(\alpha-\epsilon)e^{x_{n+1}-\wx_n}}{1-\alpha(\lambda-\epsilon)e^{x_{n+1}-\wx_n}};
\end{cases}\\
\intertext{and}
G_{k}:&
\begin{cases}[2]
e^{\epsilon p_n}=
\dfrac{1-(\alpha-\epsilon)\lambda^{-1}\left(e^{\wx_n-x_n}-1\right)}
{1-\alpha\lambda^{-1}\left(e^{\wx_n-x_n}-1\right)}\cdot
\left(1+\epsilon(\alpha+\lambda)e^{x_n-\wx_{n-1}}\right),\\
e^{\epsilon\widetilde{p}_n}=
\dfrac{1-(\alpha-\epsilon)\lambda^{-1}\left(e^{\wx_n-x_n}-1\right)}
{1-\alpha\lambda^{-1}\left(e^{\wx_n-x_n}-1\right)}
\cdot
\dfrac{1+\epsilon\alpha e^{\wx_n-\wx_{n-1}}}{1+\epsilon\alpha e^{\wx_{n+1}-\wx_n}}
\cdot
\left(1+\epsilon(\alpha+\lambda)e^{x_{n+1}-\wx_n}\right).
\end{cases}
\end{align}
The functions $\phi_{ij}(x)$, $\phi_{k\ell}(x)$, and $\phi_{i\ell}(x)$ used in Theorems~\ref{th: rtl+}, \ref{th: rtl-}, and \ref{th: rtl+-} to prove commutativity of any two of the maps $F_i$, $F_j$, $G_k$, and $G_\ell$, are given by
\begin{align*}
&\phi_{ij}(x)=\log\frac{\lambda e^x-\mu}{(\lambda-\epsilon)e^x-(\mu-\epsilon)},\qquad
\phi_{k\ell}(x)=\log\frac{(\mu+\alpha-\epsilon)e^x-(\lambda+\alpha-\epsilon)}
{(\mu+\alpha)e^x-(\lambda+\alpha)},\\
&\phi_{i\ell}(x)=\log\frac{1-\lambda(\mu+\alpha-\epsilon)e^x}{1-(\lambda-\epsilon)(\mu+\alpha)e^x}.
\end{align*}

\section{B\"acklund transformations for the additive exponential system of the relativistic Toda type}
\label{sect: asym Toda}

Performing the limit $\epsilon\to 0$ in the pluri-Lagrangian system of Section \ref{sect: master}, we arrive at the following leg functions:
\begin{align*}
&\psi_0(x)=\alpha e^x,\qquad
\psi_i(x)=\frac{1}{\lambda}(e^x-1),\qquad
\phi_{i0}(x)=\frac{(\lambda-\alpha)e^x}{1-\lambda\alpha e^x},\\
\intertext{and}
&\psi_k(x)=(\lambda+\alpha)e^x,\qquad
\phi_{k0}(x)=\frac{\lambda^{-1}(e^x-1)}{1-\lambda^{-1}\alpha(e^x-1)}.
\end{align*}
Thus, we are dealing with the following two families of B\"acklund transformations:
\begin{align}
\label{eq: Toda F}
F_{i}:&
\begin{cases}[2]
p_n=\dfrac{e^{\wx_n-x_n}-1}{\lambda}+\dfrac{(\lambda-\alpha)e^{x_n-\wx_{n-1}}}{1-\lambda\alpha e^{x_n-\wx_{n-1}}}+
\alpha e^{x_n-x_{n-1}}-\alpha e^{x_{n+1}-x_n},\\
\widetilde{p}_n=\dfrac{e^{\wx_n-x_n}-1}{\lambda}+\dfrac{(\lambda-\alpha)e^{x_{n+1}-\wx_n}}{1-\lambda\alpha e^{x_{n+1}-\wx_n}};
\end{cases}\\
\intertext{and}
\label{eq: Toda G}
G_{k}:&
\begin{cases}[2]
p_n=\dfrac{\lambda^{-1}\left(e^{\wx_n-x_n}-1\right)}{1-\alpha\lambda^{-1}\left(e^{\wx_n-x_n}-1\right)}+
(\lambda+\alpha)e^{x_n-\wx_{n-1}},\\
\widetilde{p}_n=\dfrac{\lambda^{-1}\left(e^{\wx_n-x_n}-1\right)}{1-\alpha\lambda^{-1}
\left(e^{\wx_n-x_n}-1\right)}+
(\lambda+\alpha)e^{x_{n+1}-\wx_n}-\alpha e^{\wx_{n+1}-\wx_n}+\alpha e^{\wx_n-\wx_{n-1}}.
\end{cases}
\end{align}
Any two of the symplectic maps $F_{i}$, $F_{j}$, $G_{k}$, and $G_{\ell}$ commute, which is demonstrated
in Theorems~\ref{th: rtl+}, \ref{th: rtl-}, and \ref{th: rtl+-}, with the following functions $\phi_{ij}$, $\phi_{k\ell}$, and $\phi_{i\ell}$:
\[
\phi_{ij}(x)=\frac{e^x-1}{\lambda e^x-\mu},\qquad
\phi_{k\ell}(x)=-\frac{e^x-1}{(\mu+\alpha)e^x-(\lambda+\alpha)},\qquad
\phi_{i\ell}(x)=\frac{(\lambda-\mu-\alpha) e^x}{1-\lambda(\mu+\alpha)e^x}.
\]

\section{B\"acklund transformations for the asymmetric rational system of the relativistic Toda type}
\label{sect: rat}

The last pluri-Lagrangian system we consider in the present paper is described in the following proposition.
\begin{prop}
The system of 3D-corner equations \eqref{eq: 3point Ei}, \eqref{eq: 3point Eij} with the leg functions
\begin{equation}\label{eq: legs asym rat}
\psi_j(x)=\log(x+\alpha_{j}),\qquad \phi_{ij}(x)=\log\frac{x+\beta_{j}-\beta_{i}}{x+\alpha_{j}-\alpha_{i}}
\end{equation}
is consistent, and the corresponding discrete Lagrangian 2-form $\Ell$ is closed on its solutions.
\end{prop}
\begin{proof}
The leg functions \eqref{eq: legs asym rat} are obtained from those in \eqref{eq:master legs} via the following changes of variables and transformations of parameters:
\[
x\leadsto hx,\qquad
\alpha_i\leadsto-1+h \alpha_{i},\qquad
\beta_{i}\leadsto-1+h\beta_{i},
\]
and then sending $h\to 0$.
\end{proof}

We make the following choice of parameters:
\begin{alignat*}{5}
&\alpha_0=\frac{1}{\alpha},&\qquad
&\alpha_i=0,&\qquad
&\alpha_j=0,&\qquad
&\alpha_{k}=\frac{1}{\alpha+\lambda},&\qquad
&\alpha_{\ell}=\frac{1}{\alpha+\mu},\\
&\beta_0=0,&\qquad
&\beta_i=-\frac{1}{\lambda},&\qquad
&\beta_j=-\frac{1}{\mu},&\qquad
&\beta_{k}=0,&\qquad
&\beta_{\ell}=0.
\end{alignat*}
With this choice of parameters, we find the following leg functions:
\begin{align*}
&\psi_0(x)=\log(1+\alpha x),\qquad
\psi_i(x)=\log \frac{x}{\lambda},\qquad
\phi_{i0}(x)=\log\frac{1+\lambda x}{1+\alpha x},\\
\intertext{and}
&\psi_k(x)=\log(1+(\lambda+\alpha)x),\qquad
\phi_{k0}(x)=\log\frac{x}{\lambda+\alpha(\lambda+\alpha)x}.
\end{align*}
This corresponds to the following two families of B\"acklund transformations:
\begin{align}
F_i:&
\begin{cases}[2]
e^{p_{n}}=\dfrac{\widetilde{x}_{n}-x_{n}}{\lambda}\cdot \dfrac{1+\lambda(x_{n}-\widetilde{x}_{n-1})}{1+\alpha(x_{n}-\widetilde{x}_{n-1})}\cdot
\dfrac{1+\alpha(x_{n}-x_{n-1})}{1+\alpha(x_{n+1}-x_{n})},\\
e^{\widetilde{p}_{n}}=\dfrac{\widetilde{x}_{n}-x_{n}}{\lambda}\cdot
\dfrac{1+\lambda(x_{n+1}-\widetilde{x}_{n})}{1+\alpha(x_{n+1}-\widetilde{x}_{n})}.\end{cases}\\
\intertext{and}
G_k:&
\begin{cases}[2]
e^{p_{n}}=\dfrac{\widetilde{x}_{n}-x_{n}}{\lambda+\alpha(\lambda+\alpha)(\widetilde{x}_{n}-x_{n})}\cdot
(1+(\lambda+\alpha)(x_{n}-\widetilde{x}_{n-1})),\\
e^{\widetilde{p}_{n}}=\dfrac{\widetilde{x}_{n}-x_{n}}{\lambda+\alpha(\lambda+\alpha)(\widetilde{x}_{n}-x_{n})}
\cdot
(1+(\lambda+\alpha)(x_{n+1}-\widetilde{x}_{n}))\cdot
\dfrac{1+\alpha(\widetilde{x}_{n}-\widetilde{x}_{n-1})}{1+\alpha(\widetilde{x}_{n+1}-\widetilde{x}_{n})}.
\end{cases}
\end{align}
Any two of the symplectic maps $F_{i}$, $F_{j}$, $G_{k}$, and $G_{\ell}$ commute, as follows from Theorems
\ref{th: rtl+}, \ref{th: rtl-}, and \ref{th: rtl+-} with the functions
\begin{align*}
&\phi_{ij}(x)=\log\frac{\mu-\lambda+\lambda\mu x}{x},\qquad
\phi_{k\ell}(x)=\log\frac{x}{\lambda-\mu+(\lambda+\alpha)(\mu+\alpha)x},\\
&\phi_{i\ell}(x)=\log\frac{1+\lambda x}{1+(\mu+\alpha)x}.
\end{align*}

\section{Spectrality, integrals of motion, and conservation laws}
\label{sect: spectrality}

In all our examples the maps $F_i$ and $G_k$ as introduced in Definitions \ref{def: Fi} and \ref{def: Gk} are instances of {\em B\"acklund transformations}, i.e., one-parameter families of commuting symplectic maps characterized by parameter-dependent Lagrange functions $\ELL(x,\wx;\lambda)$, resp.\ $\M(x,\wx;\lambda)$. For such families, the following remarkable property was established in \cite{S}:
\begin{theo}\label{th: spectrality}
For the pluri-Lagrangian system built by two commuting maps $F_i$, $F_j$ from one family of B\"acklund transformations (whose action is denoted by $\widetilde{\phantom{x}}$ and $\widehat{\phantom{x}}$, respectively), the discrete multi-time Lagrangian 1-form is closed on solutions if and only if the quantity $P_F(x,\wx;\lambda):=\partial\ELL(x,\wx;\lambda)/\partial \lambda$ is a common integral of motion for all $F_j$.
\end{theo}

\noindent
This is a re-formulation of the mysterious ``spectrality property'' of B\"acklund transformations discovered by Kuznetsov and Sklyanin \cite{KS}. In the specific context of relativistic Toda-type systems, we are actually dealing with two mutually commuting families of B\"acklund transformations. In this situation, the following weaker statement can be made.
\begin{theo}\label{th: spectrality weak}
For the pluri-Lagrangian system built by two commuting maps $G_k$ and $F_j$ (whose action is denoted by
$\widetilde{\phantom{x}}$ and $\widehat{\phantom{x}}$, respectively) from two mutually commuting families of B\"acklund transformations, if the discrete multi-time Lagrangian 1-form is closed on solutions, then the quantity $P_G(x,\wx;\lambda):=\partial\M(x,\wx;\lambda)/\partial \lambda$ is a common integral of motion for all $F_j$.
\end{theo}
\begin{proof}
The closedness of the discrete multi-time Lagrangian 1-form is expressed by the following formula:
\[
\M(x,\wx;\lambda)+\ELL(\wx,\widehat{\wx};\mu)-\M(\whx,\widehat{\wx};\lambda)-\ELL(x,\whx;\mu)=0.
\]
Differentiating with respect to $\lambda$ and taking into account that the terms with $\partial\wx/\partial\lambda$ etc. vanish due to the corresponding corner equations, we arrive at
\[
\frac{\partial\M(\whx,\widehat{\wx};\lambda)}{\partial\lambda}-\frac{\partial\M(x,\wx;\lambda)}{\partial\lambda}=0,
\]
which is the required property.
\end{proof}

Thus, for the maps $F_i$ and $G_k$ as introduced in Definitions \ref{def: Fi} and \ref{def: Gk}, with the understanding that the corresponding Lagrangian depend on the parameters, as in the examples discussed above, i.e.,
\begin{align}
&\ELL(x,\wx;\lambda)=\sum_n L_i(\wx_n-x_n;\lambda)-\sum_n L_0(x_{n+1}-x_n;\alpha)-
\sum_n\Lambda_{i0}(x_{n+1}-\wx_n;\lambda,\alpha),\\
&\M(x,\wx;\lambda)=\sum_n \Lambda_{k0}(\wx_n-x_n;\lambda,\alpha)+\sum_n L_0(\wx_n-\wx_{n-1};\alpha)-
\sum_n L_k(x_n-\wx_{n-1};\lambda),
\end{align}
we arrive at the following generating functions of common integrals of all the maps $F_j$, $G_\ell$:
\begin{align}\label{eq:int1}
&P_F(x,\wx;\lambda)=\sum_n\frac{\partial L_i(\wx_n-x_n;\lambda)}{\partial\lambda}-
\sum_n\frac{\partial\Lambda_{i0}(x_{n+1}-\wx_n;\lambda,\alpha)}{\partial\lambda},\\
&P_G(x,\wx;\lambda)=\sum_n\frac{\partial \Lambda_{k0}(\wx_n-x_n;\lambda,\alpha)}{\partial\lambda}-
\sum_n\frac{\partial L_k(x_n-\wx_{n-1};\lambda)}{\partial\lambda}.
\end{align}\par
We are going to demonstrate that the 2-dimensional pluri-Lagrangian interpretation allows us to get additional important information. Namely, we can derive the {\em local form} of the fact that $P_F$ and $P_G$ are the integrals of motion. For the sake of simplicity, we restrict ourselves to the local form of the statement that $P_F(x,\wx;\lambda)$ is an integral of motion of the maps $F_j$.
\begin{theo}\label{th: 2d conserv laws}
If the three-point discrete 2-form $\Ell$ is closed on solutions of the system of 3D corner equations corresponding to the maps $F_i$ and $F_j$, then the latter system admits the conservation law
\begin{equation}\label{eq: 2d conserv law}
    \Delta_j P_{i0}=\Delta_0 P_{ij},
\end{equation}
with the densities
\begin{align}
    &P_{i0} = \frac{\partial L_i(\wx_n-x_n;\lambda)}{\partial\lambda}-
    \frac{\partial\Lambda_{i0}(x_{n+1}-\wx_n;\lambda,\alpha)}{\partial\lambda},
    \label{eq: 2d conserv law dens i0}\\
    &P_{ij} = \frac{\partial L_i(\wx_n-x_n;\lambda)}{\partial\lambda}-
    \frac{\partial\Lambda_{ij}(\whx_n-\wx_n;\lambda,\mu)}{\partial\lambda}.
    \label{eq: 2d conserv law dens ij}
\end{align}
Observe that $P_{i0}$ is the summand in~\eqref{eq:int1}.\par
(Recall that $\Delta_j=T_j-I$, $\Delta_0=T_0-I$, where $T_j$ is the shift corresponding to the map $F_j$ and denoted by $\ \widehat{\ }$, while $T_0$ is the shift $n\to n+1$.)
\end{theo}
\begin{proof} We start by re-writing expression \eqref{eq: 3point S} for $d\Ell$ in the form specific for our present context:
\begin{equation}\label{eq: dL}
\begin{split}
d\Ell=S^{ij0}& = L_i(\wx_{n+1}-x_{n+1};\lambda)+L_j(\widehat{\wx}_n-\wx_n;\mu)+L_0(\whx_{n+1}-\whx_n;\alpha)\\
       &\quad -L_i(\widehat{\wx}_n-\whx_n;\lambda)-L_j(\whx_{n+1}-x_{n+1};\mu)-L_0(\wx_{n+1}-\wx_n;\alpha)\\
       &\quad -\Lambda_{ij}(\whx_{n+1}-\wx_{n+1};\lambda,\mu)
             -\Lambda_{j0}(\wx_{n+1}-\widehat{\wx}_n;\mu,\alpha)
             +\Lambda_{i0}(\whx_{n+1}-\widehat{\wx}_n;\lambda,\alpha)\\
       &\quad +\Lambda_{ij}(\whx_n-\wx_n;\lambda,\mu)+
             \Lambda_{j0}(x_{n+1}-\whx_n;\mu,\alpha)-
             \Lambda_{i0}(x_{n+1}-\wx_n;\lambda,\alpha)=0.\qquad
\end{split}
\end{equation}
Differentiating equation \eqref{eq: dL} with respect to $\lambda$ and taking into account that the terms containing $\partial \wx_n/\partial\lambda$ etc.\ vanish by virtue of the corresponding 3D corner equations, we arrive at
\begin{multline*}
\frac{\partial L_i(\wx_{n+1}-x_{n+1};\lambda)}{\partial\lambda}-
\frac{\partial L_i(\widehat{\wx}_n-\whx_n;\lambda)}{\partial\lambda}\\
 -\frac{\partial\Lambda_{ij}(\whx_{n+1}-\wx_{n+1};\lambda,\mu)}{\partial\lambda}
+\frac{\partial\Lambda_{i0}(\whx_{n+1}-\widehat{\wx}_n;\lambda,\alpha)}{\partial\lambda}\\
 +\frac{\partial\Lambda_{ij}(\whx_n-\wx_n;\lambda,\mu)}{\partial\lambda}-
             \frac{\partial\Lambda_{i0}(x_{n+1}-\wx_n;\lambda,\alpha)}{\partial\lambda}\, = \,  0.
\end{multline*}
This is equivalent to the statement of the theorem.
\end{proof}

\paragraph{Example:} Modified exponential system. We start with the maps $F_i$. An easy computation shows:
\[
\frac{\partial L_i(x;\lambda)}{\partial\lambda}=-\frac{x}{\lambda},\qquad
\frac{\partial \Lambda_{i0}(x;\lambda,\alpha)}{\partial\lambda}=
\frac{1}{\lambda}\log\left(1+\lambda e^x\right).
\]
Thus, Theorems \ref{th: spectrality} and \ref{th: spectrality weak} lead to the following generating function of integrals of motion for all maps $F_j$, $G_\ell$:
\begin{align}\notag
&\frac{\partial\ELL(x,\wx;\lambda)}{\partial \lambda}=\log P_F(x,\wx;\lambda),\\
\intertext{where}
\label{eq:spect5 mod}
&P_F(x,\wx;\lambda)=\prod_{n=1}^N e^{\wx_n-x_n}\left(1+\lambda e^{x_{n+1}-\wx_n}\right).
\end{align}
It is instructive to have a look at the local form of this result provided by Theorem \ref{th: 2d conserv laws}. We compute:
\[
\frac{\partial \Lambda_{ij}(x;\lambda,\mu)}{\partial\lambda}=
\frac{1}{\lambda}\log\left(\lambda e^x-\mu\right),
\]
and end up with the conservation law \eqref{eq: 2d conserv law} with the densities
\begin{equation}
    P_{i0}=\log e^{\wx_n-x_n}\left(1+\lambda e^{x_{n+1}-\wx_n}\right), \qquad
    P_{ij}=\log e^{\wx_n-x_n}\left(\lambda e^{\whx_n-\wx_n}-\mu\right).
\end{equation}\par

We can give a nice expressions for the generating function of integrals $P_F(x,\wx;\lambda)$ in terms of the canonically conjugate variables $x$, $p$.
\begin{theo}\label{th: mod Toda F generating function}
Set
\begin{align}
&U_n(x,p;\lambda)=
\begin{pmatrix}
1+\lambda\left( e^{p_n}+e^{x_n-x_{n-1}}\right) & -\lambda(\lambda-\alpha)e^{x_n-x_{n-1}+p_{n-1}} \\
1 & 0
\end{pmatrix},\\
\intertext{and further}\notag
&T_{N}(x,p;\lambda)=U_N(x,p;\lambda)\ldots U_2(x,p;\lambda)U_1(x,p;\lambda).
\end{align}
Then, in the periodic case $P_F(x,\wx;\lambda)$ is an eigenvalue of the matrix $T_{N}(x,p;\lambda)$, while in the open-end case $P_F(x,\wx;\lambda)$ is the (1,1)-entry of the matrix $T_{N}(x,p;\lambda)$.
\end{theo}
\begin{proof}
Set
\[
\gamma_n=e^{\wx_n-x_n}\left(1+\lambda e^{x_n-\wx_{n-1}}\right)
\frac{1+\alpha e^{x_{n+1}-\wx_n}}{1+\alpha e^{x_{n+1}-x_n}}\cdot
\frac{1+\alpha e^{x_n-x_{n-1}}}{1+\alpha e^{x_n-\wx_{n-1}}},
\]
so that $\gamma_N\cdots \gamma_2\gamma_1=P_F(x,\wx;\lambda)$ in both the periodic and the open-end cases. By
a straightforward computation based on the first formula in \eqref{eq: mod Toda F} one checks that
\begin{align*}
&U_n(x,p;\lambda)[\gamma_{n-1}]=\gamma_n,
\intertext{where we write}
&\begin{pmatrix}
a&b\\
c&d
\end{pmatrix}
[z]:=\frac{az+b}{cz+d}.
\end{align*}
Equivalently,
\[
U_n(x,p;\lambda)\begin{pmatrix} \gamma_{n-1} \\ 1 \end{pmatrix}\sim\begin{pmatrix} \gamma_n \\ 1 \end{pmatrix}.
\]
The proportionality coefficient can be determined by comparing the second components of these vectors, which results in
\begin{equation}\label{eq: proof mod Toda aux}
U_n(x,p;\lambda)\begin{pmatrix} \gamma_{n-1} \\ 1 \end{pmatrix}=
\gamma_{n-1}\begin{pmatrix} \gamma_n \\ 1 \end{pmatrix}.
\end{equation}
Thus, in the periodic case $\gamma_N\cdots \gamma_2\gamma_1$ is the eigenvalue of $T_{N}(x,p;\lambda)$ corresponding to the eigenvector $\begin{pmatrix} \gamma_0 \\ 1 \end{pmatrix}$. In the open end case, equation \eqref{eq: proof mod Toda aux} holds true for $2\le n\le N$, and has to be supplemented with
\[
U_1(x,p;\lambda)\begin{pmatrix} 1 \\ 0 \end{pmatrix}=\begin{pmatrix} \gamma_1 \\ 1 \end{pmatrix}\quad \text{and}\quad \begin{pmatrix} 1 & 0 \end{pmatrix}\begin{pmatrix} \gamma_N \\ 1 \end{pmatrix}=\gamma_N.
\]
This yields
\[
\begin{pmatrix} 1 & 0 \end{pmatrix} T_{N}(x,p;\lambda)\begin{pmatrix} 1 \\ 0 \end{pmatrix}=\gamma_N\cdots\gamma_2\gamma_1.
\]
\end{proof}

Turning now to the maps $G_k$, we find:
\begin{align}
& \frac{\partial \Lambda_{k0}(x;\lambda,\alpha)}{\partial\lambda}=
-\frac{1}{\lambda+\alpha}x+
\frac{1}{\lambda+\alpha}\log\left(1-\frac{\alpha}{\lambda}\big(e^x-1\big)\right),\\
& \frac{\partial L_k(x;\lambda)}{\partial\lambda}=
\frac{1}{\lambda+\alpha}\log\left(1+(\lambda+\alpha)e^x\right).
\end{align}
According to Theorems \ref{th: spectrality} and \ref{th: spectrality weak}, we compute:
\begin{align}\notag
&\frac{\partial\M(x,\wx;\lambda)}{\partial\lambda}=-\frac{1}{\lambda+\alpha}\log P_G(x,\wx;\lambda),\\
\intertext{where}
\label{eq:spect6 mod Toda}
&P_G(x,\wx;\lambda)=\prod_{n=1}^{N}\frac{e^{\wx_n-x_n}}{1-\dfrac{\alpha}{\lambda}\left(e^{\wx_n-x_n}-1\right)}
\left(1+(\lambda+\alpha)e^{x_n-\wx_{n-1}}\right).
\end{align}
Again,  $P_G(x,\wx;\lambda)$ is a generating function of common integrals of motion for all the maps $F_j$, $G_\ell$. Remarkably, in terms of the canonically conjugate variables $x$, $p$ this is essentially the same function as $P_F(x,\wx;\lambda)$.
\begin{theo}\label{th: mod Toda G generating function}
In the periodic case, $P_G(x,\wx;\lambda)$ is an eigenvalue of the matrix $T_{N}(x,p;\lambda+\alpha)$, while in the open-end case $P_G(x,\wx;\lambda)$ is the (1,1)-entry of the matrix $T_{N}(x,p;\lambda+\alpha)$.
\end{theo}
\begin{proof}
Set
\[
\beta_n=\frac{e^{\wx_n-x_n}}{1-\dfrac{\alpha}{\lambda}\left(e^{\wx_n-x_n}-1\right)}
\left(1+(\lambda+\alpha)e^{x_n-\wx_{n-1}}\right).
\]
Then one easily checks with the help of the first formula in \eqref{eq: mod Toda G} that
\[
U_n(x,p;\lambda+\alpha)[\beta_{n-1}]=\beta_n.
\]
Then, the proof goes along the same lines as the proof of Theorem \ref{th: mod Toda F generating function}.
\end{proof}

\section{Conclusions}
In the present paper, we applied the general theory of two-dimensional pluri-Lagrangian systems, developed in \cite{variational}, to the analysis of B\"acklund transformations for relativistic Toda-type systems. It was possible due to a novel way to embed the one-dimensional relativistic Toda-type systems into certain two-di\-men\-sion\-al lattice systems. This embedding is well suited for a proof of commutativity of B\"acklund transformations, as well as for the proof of the closure relation of the corresponding action functional.
A different relation of this kind between relativistic Toda-type systems and 3D consistent systems of quad-equations was discussed previously, cf.~\cite{quadgraphs,Q4,Todapaper}, and led to a systematic derivation of zero-curvature representations for (discrete and continuous) relativistic and non-relativistic Toda-type systems. A connection between these two approaches is presently unclear and remains to be worked out in detail.

\section*{Acknowledgments}
This research was supported by the DFG Collaborative Research Center TRR 109 ``Discretization in Geometry and Dynamics''.

{\small
\bibliographystyle{amsalpha}

\begin{thebibliography}{YKLN11}

\bibitem[Adl99]{A}
Vsevolod~E. Adler,
   \emph{{ Legendre transformations on the triangular lattice}},
    Funct. Anal. Appl. \textbf{34} (1999), pp. 1--9.

\bibitem[ABS03]{ABS1}
Vsevolod~E. Adler, Alexander~I. Bobenko, and Yuri~B. Suris,
  \emph{{Classification of Integrable Equations on Quad-Graphs. The Consistency
  Approach}}, Comm. Math. Phys. \textbf{233} (2003), pp. 513--543.

\bibitem[AS97a]{ASh1}
Vsevolod~E. Adler and Alexei~B. Shabat,
   \emph{{On a class of Toda chains}},
   Theor. Math. Phys. \textbf{111} (1997), pp. 647--657.

\bibitem[AS97b]{ASh2}
\bysame,
   \emph{{Generalized Legendre transformations}},
    Theor. Math. Phys. \textbf{112} (1997), pp. 935--948.

\bibitem[ALN12]{ALN}
James Atkinson, Sarah Lobb, and Frank~W. Nijhoff, \emph{{An integrable
  multicomponent quad equation and its Lagrangian formulation}}, Theor. Math.
  Phys. \textbf{173} (2012), pp. 1644--53.

\bibitem[AS04]{Q4}
Vsevolod~E. Adler and Yuri~B. Suris, \emph{Q4: integrable master equation
  related to an elliptic curve}, Intern. Math. Research Notices (2004), no.~47,
  pp. 2523--2553.

\bibitem[BS02]{quadgraphs}
Alexander~I. Bobenko and Yuri~B. Suris, \emph{Integrable systems on
  quad-graphs}, Intern. Math. Research Notices \textbf{11} (2002), pp.
  573--611.

\bibitem[BS08]{DDG}
\bysame, \emph{{Discrete differential geometry. Integrable Struture}}, Graduate
  Studies in Mathematics, vol.~98, AMS, 2008.

\bibitem[BS10a]{BS1}
\bysame, \emph{{On the Lagrangian structure of integrable quad-equations}},
  Lett. Math. Phys. \textbf{92} (2010), no.~1, pp. 17--31.

\bibitem[BS14]{pluriharmonic}
\bysame, \emph{{Discrete pluriharmonic functions as solutions of linear pluri-Lagrangian systems}},
  arXiv:1403.2876 [math-ph] (2014), to appear in Commun. Math. Phys.

\bibitem[BPS13]{nonrel}
Raphael Boll, Matteo Petrera, and Yuri~B. Suris, \emph{{Multi-time Lagrangian
  1-forms for families of B{\"a}cklund transformations. Toda-type systems}}, J.
  Phys. A: Math. Theor. \textbf{46} (2013), no.~275204.


\bibitem[BPS14a]{variational}
\bysame, \emph{{What is integrability of discrete variational systems?}}, Proc.
  R. Soc. A \textbf{470} (2014), no.~20130550.

\bibitem[BPS14b]{octahedron}
\bysame, \emph{{On integrability of discrete variational systems. Octahedron
  relations}}, arXiv:1406.0741 [nlin.SI] (2014).

\bibitem[BS10b]{Todapaper}
Raphael Boll and Yuri~B. Suris, \emph{{Non-symmetric discrete Toda systems from
  quad-graphs}}, Applicable Analysis \textbf{89} (2010), no.~4, pp. 547--569.

\bibitem[BS12]{lagrangian}
\bysame, \emph{{On Lagrangian structures of 3D consistent systems of asymmetric
  quad-equations}}, J. Phys. A: Math. Theor. \textbf{45} (2012), no.~115201.


\bibitem[BFPP93]{BFPP}
Fran Burstall, Dirk Ferus, Franz Pedit, and Ulrich Pinkall, \emph{{Harmonic
  tori in symmetric spaces and commuting Hamiltonian systems on loop
  algebras}}, Ann. Math. \textbf{138} (1993), pp. 173--212.

\bibitem[KS98]{KS}
Vadim~B. Kuznetsov and Evgeny~K. Sklyanin, \emph{{On B{\"a}cklund
  transformations for many-body systems}}, J. Phys. A: Math. Gen. \textbf{31}
  (1998), pp. 2241--51.

\bibitem[LN09]{LN}
Sarah Lobb and Frank~W. Nijhoff, \emph{Lagrangian multiforms and
  multidimensional consistency}, J. Phys. A: Math. Theor. \textbf{42} (2009),
  no.~454013.

\bibitem[LN10]{LN2}
\bysame, \emph{{Lagrangian multiform structure for the lattice Gel'fand-Dikii
  hierarchy}}, J. Phys. A: Math. Theor. \textbf{43} (2010), no.~072003.

\bibitem[LNQ09]{LNQ}
Sarah Lobb, Frank~W. Nijhoff, and G.~R.~W. Quispel, \emph{{Lagrangian multiform
  structure for the lattice KP system}}, J. Phys. A: Math. Theor. \textbf{42}
  (2009), no.~472002.

\bibitem[Nij02]{N}
Frank~W. Nijhoff, \emph{{Lax pair for the Adler (lattice Krichever-Novikov)
  system}}, Phys. Lett. A \textbf{297} (2002), pp. 49--58.

\bibitem[OV90]{OV}
Yoshihiro Ohnita and Giorgio Valli, \emph{{Pluriharmonic maps into compact Lie
  groups and factorization into unitons}}, Proc. London Math. Soc. \textbf{61}
  (1990), pp. 235--257.

\bibitem[Rud69]{R}
Walter Rudin, \emph{Function theory in polydiscs}, Benjamin, 1969.

\bibitem[Sur03]{S}
Yuri~B. Suris, \emph{{The problem of discretization: Hamiltonian approach}},
  Progress in mathematics, vol. 219, Birkh{\"a}user Verlag, Basel, 2003.

\bibitem[Sur13]{S12}
\bysame, \emph{{Variational formulation of commuting Hamiltonian flows:
  multi-Lagrangian 1-forms}}, J. Geom. Mech. \textbf{5} (2013), pp. 365--379.


\bibitem[YKLN11]{YLN}
Sikarin Yoo-Kong, Sarah Lobb, and Frank~W. Nijhoff, \emph{{Discrete-time
  Calogero-Moser system and Lagrangian 1-form structure}}, J. Phys. A: Math.
  Theor. \textbf{44} (2011), no.~365203.

\end{thebibliography}
\providecommand{\bysame}{\leavevmode\hbox to3em{\hrulefill}\thinspace}
\providecommand{\MR}{\relax\ifhmode\unskip\space\fi MR }
\providecommand{\MRhref}[2]{%
  \href{http://www.ams.org/mathscinet-getitem?mr=#1}{#2}
}
\providecommand{\href}[2]{#2}

}

\end{document}